\documentclass[12pt]{article}

\usepackage{jheppub,todonotes}
\usepackage{amsmath,amssymb,amsfonts,amsthm}
\usepackage{mathtools}
\usepackage{revsymb}
\usepackage{bm}
\usepackage{tikz}
\usepackage{tikz-cd}
\usetikzlibrary{matrix,arrows,automata}
\usetikzlibrary{automata}
\usepackage{graphicx}
\usepackage{soul}
\usepackage{enumitem}
\usepackage{subfigure}

\makeatletter
\def\@fpheader{\relax}
\makeatother

\protected\def\betterat{{\fontfamily{ptm}\selectfont @}}
\begingroup\lccode`~=`@ \lowercase{\endgroup\let~}\betterat
\catcode`@=\active

\renewcommand{\baselinestretch}{1.2}

\newcommand\be{\begin{equation}}
\newcommand\ee{\end{equation}}
\newcommand\beq{\begin{equation}}
\newcommand\eeq{\end{equation}}
\newcommand\bea{\begin{eqnarray}}
\newcommand\eea{\end{eqnarray}}
\newcommand\ba{\begin{array}}
\newcommand\ea{\end{array}}

\newcommand\eref[1]{(\ref{#1})}
\newcommand\comment[1]{}
\newcommand\pa{\partial}
\renewcommand\tilde{\widetilde}

\newcommand\eps{\epsilon^{\alpha \beta}}
\newcommand\barH{\overline{H}}
\newcommand\barF{\overline{F}}

\newcommand{\tu}{\tilde{u}}
\newcommand{\td}{\tilde{d}}
\newcommand{\te}{\tilde{e}}

\newtheorem{thm}{Theorem}[section]
\newtheorem*{thm*}{Theorem}

\newtheorem{prop}[thm]{Proposition}

\newtheorem*{prop*}{Proposition}

\newtheorem{lem}[thm]{Lemma}

\newtheorem{exm}[thm]{Example}
\newtheorem{comp}[thm]{Computation}

\newcommand{\PP}{{\mathbb P}}

\newcommand{\CC}{{\mathbb C}}
\newcommand{\ZZ}{{\mathbb Z}}

\newcommand{\cV}{{\mathcal V}}

\def\eps{\epsilon^{\alpha \beta}}
\def\barH{\overline{H}}

\begin{document}

\title{The Vacuum Moduli Space of the\\ Minimal Supersymmetric Standard Model}
\author[*,\dagger]{Yang-Hui He,}
\author[\sharp]{Vishnu Jejjala,}
\author[\flat]{Brent D.\ Nelson,}
\author[\times,\star]{Hal Schenck,}
\author[\circ,\star]{Michael Stillman}

\affiliation[*]{London Institute for Mathematical Sciences, Royal Institution, London, W1S 4BS, UK}
\affiliation[\dagger]{Merton College, University of Oxford, OX1 4JD, UK}
\affiliation[\sharp]{Mandelstam Institute for Theoretical Physics, School of Physics, and NITheCS, University of the Witwatersrand, Johannesburg, WITS 2050, South Africa}
\affiliation[\flat]{Department of Physics, Northeastern University, Boston, MA 02115, USA}
\affiliation[\times]{Department of Mathematics, Auburn University, Auburn, AL 36849, USA}
\affiliation[\star]{Mathematical Institute, University of Oxford, OX2 6GG, UK}
\affiliation[\circ]{Department of Mathematics, Cornell University, Ithaca, NY 14853, USA\\}

\emailAdd{hey@maths.ox.ac.uk}
\emailAdd{v.jejjala@wits.ac.za}
\emailAdd{b.nelson@northeastern.edu}
\emailAdd{hks0015@auburn.edu}
\emailAdd{mike@math.cornell.edu}


\begin{abstract}{
A starting point in the study of the minimal supersymmetric Standard Model (MSSM) is the vacuum moduli space, which is a highly complicated algebraic variety:
it is the image of an affine variety $X \subset \CC^{49}$ under a symplectic quotient map $\phi$ to $\CC^{973}$.
Previous work~\cite{He:2014oha} computed the vacuum moduli space of the electroweak sector; geometrically this corresponds to studying a restriction of $\phi$:
$\CC^{13} \stackrel{\phi^{\texttt{res}}}{\longrightarrow} \CC^{22}$.
We analyze the geometry of the full vacuum moduli space for superpotentials $W_{\rm minimal}$ (without neutrinos) and $W_{\rm MSSM}$ (with neutrinos) in $\mathbb{C}^{973}$. 
In both cases, we prove that $X$ consists of three irreducible components $X_1$, $X_2$, and $X_3$, and determine the images $M_i$ of the $X_i$ under $\phi$. For $W_{\rm minimal}$ we show they have, respectively, dimensions $1$, $15$, and $29$, and prove that each of the $M_i$ is a rational variety, while for 
$W_{\rm MSSM}$ we show that $M_3$ is the only component. 
 Restricting the $M_i$ to the electroweak sector, we recover the results in~\cite{He:2014oha}.  We describe the components of the vacuum moduli space geometrically in terms of incidence varieties to a product of Segre varieties.} 
\end{abstract}

\setlength{\parindent}{15pt}
\setlength{\parskip}{10pt}
\renewcommand{\baselinestretch}{1.15}

\maketitle

\newpage

\section{Introduction}
The Standard Model of particle physics~\cite{Glashow:1961tr,Weinberg:1967tq,Salam:1968rm,Glashow:1970gm} is the most precisely tested theory in the history of science~\cite{Morel:2020dww}.
It explains the electromagnetic, strong, and weak interactions as well as the Higgs mechanism and classifies the known elementary particles~\cite{Burgess:2006hbd,Schwartz:2014sze,Hubsch:2015vpa}.
Though recent collider experiments hint at intriguing mismatches between theory and observation, to date there is no $5\sigma$ discovery of particle physics beyond the Standard Model~\cite{ParticleDataGroup:2024cfk}.
We know, however, that new physics is compulsory.
For example, the Standard Model does not explain what dark matter is, the mechanism by which neutrinos acquire mass and what these masses are, or the requisite CP violation necessary to accommodate a matter/antimatter asymmetry in the early Universe~\cite{Ellis:2009pz,Langacker:2017uah}.
We must also ultimately reconcile quantum field theory with gravitation~\cite{Green:1987sp}.

\subsection{Supersymmetry} We wish to understand not only how particle physics is organized but also whence this structure originates.
The best understood top down constructions of Standard Model-like physics arise from Calabi--Yau compactifications of ten-dimensional superstring theory down to a four-dimensional effective theory in $\mathbb{R}^{1,3}$~\cite{candelas1985vacuum,Green:1987mn}.
By virtue of having a Ricci-flat metric on the Calabi--Yau space, such heterotic compactifications enjoy $\mathcal{N}=1$ supersymmetry in four dimensions and are source free, leading order solutions to the Einstein equation.
Supersymmetry posits a mass degeneracy between fermionic and bosonic degrees of freedom.
(See~\cite{Bailin:1994qt,Weinberg:2000cr} for reviews.)
This is not seen empirically in Nature: there is not, for instance, a scalar particle with the same mass as the electron.
At best, supersymmetry is realized at low energies as a softly broken symmetry.
Nevertheless, it is a theoretically appealing idea: it explicates the hierarchy between the electroweak symmetry breaking scale and the Planck mass, it allows for the unification of gauge couplings, it supplies candidates for weakly interacting dark matter, and it can enhance the amount of CP violation within the Standard Model.
The expectation of detecting superpartners of Standard Model particles at the Large Hadron Collider is so far sadly dashed by data.
There is no experimental corroboration of supersymmetry at scales below $1$ TeV~\cite{Canepa:2019hph,ParticleDataGroup:2024cfk}.
Optimistically, perhaps this only means that parameters in a supersymmetric Standard Model are somewhat finely tuned and the superpartners are accessible at a higher energy regime to be probed by the next generation of particle accelerators.
We are agnostic about this possibility.

Because it is the simplest theory consistent with observation and compatible with the modern framework for quantum gravity, we adopt the view that the minimal supersymmetric Standard Model (MSSM)~\cite{Dimopoulos:1981zb} remains important to study in its own right.
Since cosmological measurements indicate that our Universe today most closely resembles de Sitter space~\cite{SupernovaSearchTeam:1998fmf,SupernovaCosmologyProject:1998vns}, which is incompatible with supersymmetry~\cite{Gibbons:1982fy,Freedman:2012zz}, we of course view this analysis as an initial step.
The MSSM is a modestly unreal world, but its physics is not excluded by data.

To investigate a supersymmetric gauge theory, we must examine its matter content and interactions.
The superpotential $W(\Phi)$ is a holomorphic function of the superfields that encodes this information.
The vacuum is the simultaneous solution to the F-flatness conditions obtained from taking derivatives of the superpotential with respect to the scalar components of the chiral superfields and the D-flatness conditions associated to the generators of the gauge group, which for the MSSM is $G = SU(3)_C\times SU(2)_L\times U(1)_Y$.\footnote{
In this work, we do not differentiate between $G$ and $G/N$, where $N$ is a subgroup of $\mathbb{Z}_6$.
This subtlety distinguishes different grand unified theories~\cite{Georgi:1974sy}.
}
This describes the set of zero energy states modulo gauge transformations.
In general, the field configurations defined by scalar vacuum expectation values (vevs) for which the potential determined by the F-terms and D-terms vanishes does not specify a unique point in field space.
Rather, there is a \textit{vacuum moduli space} of solutions~\cite{Witten:1981nf}.\footnote{
It should be emphasized that throughout we only consider the so-called Higgs branch of the vacuum moduli space.
The Coulomb branch~\cite{Intriligator:1995au}, on which there has been much recent activity, will be left to a separate investigation.}
This is a manifold $\mathcal{M}$, conveniently described in algebraic geometry.
Moreover, since the superpotential does not acquire quantum corrections in perturbation theory, the vacuum moduli space $\mathcal{M}$ is a robust and defining feature of the theory.\footnote{
The superpotential does receive non-perturbative corrections, \textit{e.g.}, from instantons.
This means that the quantum moduli space can be different from the classical moduli space.
SQCD with $N_c = N_f$ is an example of a theory that exhibits this phenomenon: the origin in the classical moduli space is excised in the true, quantum moduli space~\cite{Seiberg:1994bz}.}
The gauge invariant operators (GIOs) of a supersymmetric field theory provide coordinates on $\mathcal{M}$.
In this work, we complete the program initiated in~\cite{Gray:2005sr,Gray:2006jb,Gray:2008yu,He:2014oha} and explicitly solve for the full classical vacuum moduli space of the MSSM with a renormalizable superpotential.

\subsection{Dramatis person\ae} 
Throughout this paper, we will use the standard notation that the matter content and the superpotential (taken to be the minimal renormalizable one) of the MSSM are:
\[
\begin{array}{cc}
\mbox{
\begin{tabular}{|c|c|}\hline
$Q_{a, \alpha}^i$ & doublet quarks \\ \hline
$u_a^i$ & singlet up-type quarks \\ \hline
$d_a^i$ & singlet down-type quarks \\ \hline
$L_{\alpha}^i$ & doublet leptons \\ \hline
$e^i$ & singlet leptons \\ \hline
$\nu^i$ & neutrinos \\ \hline
$H_\alpha$ & up-type Higgs \\ \hline
$\barH_\alpha$ & down-type Higgs \\ \hline
\end{tabular}
}
& \quad
\begin{array}{rcl}
W_{\rm MSSM} &=& C^0 H_\alpha \barH_\beta \eps + C^1_{ij} Q^i_{a,\alpha} u^j_a H_\beta \eps \\
&& +\, C^2_{ij} Q^i_{a,\alpha} d^j_a \barH_\beta \eps + C^3_{ij} e^i L^j_{\alpha} \barH_\beta \eps \\
&& +\, C^4_{ij} \nu^i \nu^j + C^5_{ij} \nu^i L^j_{\alpha} H_\beta \eps ~,\\
\text{and }&&\\
&&\\
W_{\rm minimal} & = & W_{\rm MSSM} \mbox{ with }\nu^i \mbox{ set to }0 ~,
\end{array}
\end{array}
\]
where $i = 1,2,3$ (corresponding to generation or family), $a = 1,2,3$ (corresponding to color), and $\alpha,\beta = 1,2$ (corresponding to the weak interaction).
In the superpotential, the coupling constants $C^A$ are taken to be generic (complex) numbers.
\vspace{-.1in}
\subsection{Outline of paper and main results}
\vspace{-.1in}
Much of the introduction should be familiar to high energy physicists.
The context for why solving for the vacuum moduli space of the MSSM is an interesting problem to study is largely directed to algebraic geometers.
The organization of the remainder of the paper is as follows.
\vspace{-.1in}
\begin{itemize}
    \item In \S\ref{sec:preliminaries}, we describe an algorithm for finding the vacuum moduli space of a four-dimensional $\mathcal{N}=1$ supersymmetric field theory in terms of the fields and GIOs. In order to apply this algorithm and to set notation, we present the MSSM matter content and GIOs. We close \S\ref{sec:preliminaries} with a discussion of previous work on the problem. In parallel to the introduction addressed to algebraic geometers, this section couches the necessary mathematics in the language of $\mathcal{N}=1$ quantum field theory in a manner that we hope is accessible to particle theorists.
    \item In \S\ref{sec:three}, we analyze the master space of the MSSM, which is a complex affine algebraic variety $X \subset \CC^{49}$. Algebraically it corresponds to the Jacobian ideal of the superpotential $W$; geometrically it defines the singular locus of the hypersurface defined by the vanishing of $W$. We prove $X$ consists of three irreducible components $X_1, X_2,$ and $X_3$, and give explicit defining equations for each component. 
    \item In \S\ref{sec:four}, we begin our analysis of the vacuum moduli space (henceforth abbreviated VMS). We denote by $M_i$ the image of $X_i$ under the symplectic quotient map. We show that both $X_1$ and $X_2$ map to relatively small subspaces of $\CC^{973}$ (to $\CC^3$ and $\CC^{72}$, respectively). We show that $M_1$ is a line, and that $M_2$ is birational to a variety in $\CC^{39}$ which is an incidence correspondence between Segre varieties. Restriction to the electroweak sector corresponds algebraically to setting $Q=u=d=0$, and we show the restriction of $M_2$ yields exactly the result in~\cite{He:2014oha}.
    \item In \S\ref{sec:five}, we tackle the largest component of the master space, $X_3$. This component is of dimension $41$ in $\CC^{49}$, and maps to a $783$ dimensional subspace of the target space $\CC^{973}$; it is a complete intersection defined by four linear forms and four quadrics. We prove that the image $M_3$ is a rational variety of dimension $29$. To do this, we factor the map given by the $13$ types of non-vanishing GIOs on $X_3$ into individual steps, and construct a tower of projections terminating in a $29$ dimensional base variety in $\CC^{63}$, for which we give explicit defining equations. 
    \item In \S\ref{sec:six}, we determine the effect of adding neutrinos to the superpotential, and the concomitant changes on the master space and vacuum moduli space.  
    \end{itemize}
    
\begin{thm}
The master space $X$ and the vacuum moduli space $\mathcal{M}$ of $W_{\rm minimal}$ are complex affine algebraic varieties. They both have three irreducible components, which intersect as below. Numbers indicate (affine) dimensions of components.
\vspace{-.05in}
  \begin{figure}[h]
    \centering
    \hfill
    \subfigure{
        \begin{tikzpicture}[thick]]
            \node[label={$X_1$}] (X_1) at (-2,1) {};
            \node[label={$X_2$}] (X_2) at (4,1) {};
            \node[label={$X_3$}] (X_3) at (1,-5) {};
    
            \node[label={$23$}] (23) at (-1,0) {};
            \node[label={$27$}] (27) at (3,0) {};
            \node[label={$41$}] (41) at (1,-3.5) {};
    
            \node[label={$22$}] (22) at (1,0.25) {};
            \node[label={$22$}] (22) at (-0.25,-1.75) {};
            \node[label={$26$}] (26) at (2.25,-1.75) {};
    
            \node[label={$21$}] (21) at (1,-1.25) {};
    
            \draw (0,0) circle(2cm);
            \draw (2,0) circle(2cm);
            \draw (1,-2) circle(2cm);
        \end{tikzpicture}
    }
    \hfill
    \hfill
    \hfill
    \subfigure{
        \begin{tikzpicture}[thick]
           \node[label={$M_1$}] (M_1) at (-.25,-.75) {};
            \node[label={$M_2$}] (M_2) at (3.25,2.75) {};
            \node[label={$M_3$}] (M_3) at (0,-3) {};

            \node[label={$1$}] (1) at (.35,-.075) {};
            \node[label={$15$}] (15) at (1.75,1.5) {};
    
            \node[label={$14$}] (14) at (1,0.65) {};
            
            \node[label={$29$}] (29) at (0,-1.55) {};
            
            \draw (0,0) circle(2cm);
            \draw (1.75,1.75) circle(1.5cm);
            \draw (.35,.35) circle(0.5cm);
        \end{tikzpicture}
    }
    \hfill
    \vspace{-.2in}
     \caption{\textsf{The components $X_i$ and $M_i$ and dimensions of intersections.}}
    \label{fig:conifold}
\end{figure}

\vspace{-.15in}
\noindent We determine explicit defining equations for birational models of each of these varieties, and describe the corresponding birational models geometrically:
\vspace{-.1in}
\begin{itemize}[leftmargin=.5in,itemsep=0in]
\item $M_1$ is a line.
\item $M_2$ is a rational variety of dimension $15$. 
\item $M_3$ is a rational variety of dimension $29$.
\end{itemize}

\vspace{-.15in}
\noindent In \S\ref{sec:six}, we prove for $W_{\rm MSSM}$ that the vacuum moduli space 
consists of one irreducible component; 
the vanishing ideal is the ideal of $M_3$ along with the neutrino variables.

\noindent For both $W_{\rm MSSM}$ and $W_{\rm minimal}$, we show that the restriction to  $Q=u=d=0$ recovers previous results of~\cite{He:2014oha} on the electroweak sector of the vacuum moduli space.
\end{thm}
\section{Preliminaries}\label{sec:preliminaries}
\subsection{Defining the vacuum moduli space}\label{sec:vms}
\vskip -.1in To set notation, we briefly review textbook material~\cite{Wess:1992cp,Bailin:1994qt,Weinberg:2000cr,Argyres:2001eva,Terning:2006bq}.
The action of an $\mathcal{N}=1$ supersymmetric quantum field theory in four dimensions is
\begin{align}
\nonumber
S &=& \int d^4x\ \left[ \int d^4\theta\ K(\Phi_i^\dagger, \Phi_i) + \left( \frac{1}{4g^2} \left( \int d^2\theta\ \text{tr}\, \mathcal{W}_\alpha \mathcal{W}^\alpha +
\int d^2\overline{\theta}\ \text{tr}\,
\overline{\mathcal{W}}_\alpha \overline{\mathcal{W}}^\alpha
\right) \right. \right. \\
&& 
\left. \left.
+ \int d^2\theta\ W(\Phi_i) + \int d^2\overline{\theta}\ \overline{W}(\Phi_i^\dagger)  \right) \right] ~. \label{eq:action}
\end{align}
The $\Phi_i$ are chiral superfields, which transform in  representations $R_i$ of the gauge group $G$.
We define the gauge field strength from the chiral spinor superfield
\be
\mathcal{W}_\alpha = i \overline{D}^2 e^{-V} D_\alpha e^V ~,
\ee
with $V$ a vector superfield taking values in the Lie algebra $\mathfrak{g}$.
The K\"ahler potential includes the kinetic terms:
\be
K(\Phi_i^\dagger, \Phi_i) = \Phi^\dagger_i e^V \Phi_i + \ldots ~,
\ee
and the superpotential is a holomorphic function of the chiral superfields that describes interactions.
The integrals are over spacetime and superspace coordinates.
The action~\eref{eq:action} is invariant under the transformations
\be
\Phi\mapsto h\Phi ~, \qquad e^V\mapsto (h^{-1})^\dagger e^V h^{-1} ~, 
\ee
where $h=e^{i\Lambda}$ and $\Lambda$ is a chiral superfield.

The potential of the theory is written in terms of the lowest components of the chiral superfields:
\be
V(\phi,\phi^*) = \sum_i |F^i|^2 + \frac14 \sum_a (D^a)^2 ~. \label{eq:pot}
\ee
The F-terms and D-terms are, respectively,
\bea    
F^i &=& \frac{\pa W}{\pa \phi_i} ~, \label{eq:fterms} \\
D^a &=& g \sum_i \phi_i^* T^a \phi_i ~, \label{eq:dterms}
\eea
where the $T^a$, $a = 1,\ldots, \dim G$, are the generators of $\mathfrak{g}$.\footnote{
We have neglected the Fayet--Iliopoulos D-terms~\cite{Fayet:1974jb} associated to $U(1)$ factors, but these can be included as well.}
There is an F-term for every chiral superfield.
Supersymmetry imposes the condition that the eigenvalues of the Hamiltonian operator are non-negative.
Thus, the minimum energy states --- \textit{viz.}, the \textit{vacua} --- have zero energy, which means the potential must be zero.
Since each term in~\eref{eq:pot} must individually vanish, vacua are defined by the field configurations for which $F^i=0$ and $D^a=0$, for all $i$ and for all $a$.

The vacuum is generically a manifold rather than an isolated point.\footnote{
This is true even within the Standard Model itself.
In the Abelian Higgs mechanism, the renormalizable potential in four dimensions is quartic: $V(\phi) = \lambda( |\phi|^2 - \phi_0^2)^2$ and has the shape of a Mexican hat or the bottom of a wine bottle.
Its minimum is achieved at $\phi = e^{i\theta} \phi_0$.
The vacuum geometry is $S^1$, corresponding to the phase described by $\theta$.
Spontaneous symmetry breaking selects a particular phase and breaks the $SU(2)_L\times U(1)_Y$ electroweak symmetry to the $U(1)_{EM}$ of electromagnetism~\cite{Englert:1964et,Higgs:1964pj,Guralnik:1964eu}.
This trivializes the vacuum to a point.
}
From a physical perspective, at a particular point in the vacuum geometry, the scalar components of some of the chiral superfields may have non-zero vevs.
The gauge groups under which such a field is charged are (at least partially) broken at that point.

To solve for the vacuum, we shall work with $G^c$ the complexification of the gauge group.
The complexification of a real compact Lie group $G$ is a complex Lie group $G^c$ containing $G$ as a real subgroup such that there is a Lie algebra isomorphism between the two.
In particular, the complexification of the group $SU(N)$ is $SL(N,\mathbb{C})$, and the complexification of $U(1)$ is $\mathbb{C}^*$.
A theorem rediscovered several times~\cite{Buccella:1982nx,Procesi:1985hr,Gatto:1986bt,Luty:1995sd} establishes that the (classical) \textit{vacuum moduli space} is the symplectic quotient of the master space, which is the manifold of scalar field vevs that satisfy the F-term equations:
\be
{\cal M} = {\cal F}/\!/G = {\cal F}/G^c ~. \label{eq:vms}
\ee

To explicate this statement, consider a polynomial ring defined by the chiral superfields:
\be
A = \mathbb{C}[\Phi_1,\ldots,\Phi_m] ~.
\ee
We define the \textit{master space}~\cite{Forcella:2008bb,Forcella:2008eh} as the space obtained by quotienting this ring by the ideal of the F-term constraints:
\be
\mathcal{F} = A / X ~, \qquad X = \Big\langle \pa_i W := \frac{\pa W}{\pa \phi_i} = 0 \Big\rangle ~, \quad i=1,\ldots,m ~.
\ee
Luty and Taylor~\cite{Luty:1995sd} show that for every solution to the F-terms, there is a solution to the D-terms in the completion of the orbit of the complexified gauge group $G^c$. 
The \textit{gauge invariant operators} (GIOs) supply a basis for the D-orbits; if the basis has $n$ elements, this allows us to define a map from the master space to $\mathbb{C}^n$. To determine the image of the map (a geometric object) we introduce a second polynomial ring
\be
B = \mathbb{C}[O_1,\ldots,O_n] ~,
\ee
where the $\{O_j\}$ are a minimal basis of $G$-invariants, which are composite operators constructed from the elementary fields $\Phi_i$. Then the vacuum moduli space $\mathcal{M}$ is the vanishing locus of the polynomials in the kernel of the ring map:
\be
\pi: B \to \mathcal{F} ~.
\ee

\subsection{An algorithm}\label{sec:algorithm}

A mathematical implementation of this prescription is straightforward~\cite{Gray:2009fy,Hauenstein:2012xs}.
\begin{enumerate}\setcounter{enumi}{-1}
\item The input is the action $S$, from which we read off the chiral superfields $\Phi_i$, the gauge group $G$, and the superpotential interactions in $W$.
\item Knowing the representations $R_i$ under which the $\Phi_i$ transform, we construct a minimal list of GIOs, $\{O_j\}$.
Here, $i=1,\ldots,m$ and $j=1,\ldots,n$.
\item Define the polynomial ring
\be
R = \mathbb{C}[\phi_1,\ldots,\phi_m,y_1,\ldots,y_n] ~,
\ee
where the $y_j$ are auxiliary variables.
\item Construct the ideal
\be
J = \big\langle \pa_i W, y_j - O_j \big\rangle ~.
\ee
\item Using Gr\"obner bases, eliminate the variables $\phi_i$ from $J \subset R$.
This leaves the ideal $I_\mathcal{M} \subset \mathbb{C}[y_1,\ldots,y_n]$.
This is the ideal corresponding to the vacuum moduli space, $\mathcal{M}$.
\end{enumerate}

An example illustrates the central idea.
\comment{
Consider the $\mathcal{N}=4$ $SU(N)$ super-Yang--Mills gauge theory, which is written in $\mathcal{N}=1$ language in terms of three adjoint chiral fields with superpotential
\be
W = [\phi_1,\phi_2]\phi_3 ~.
\ee
The F-term constraints tell us that
\be
[\phi_i,\phi_j] = 0 ~,
\ee
for any pair $i\ne j$.
Thus, the $\phi_i$, which can be regarded as $N\times N$ traceless matrices, are simultaneously diagonalizable.
That is to say, they live in the Cartan subalgebra.
The vacuum moduli space has real dimension $6(N-1)$, corresponding to the $N-1$ complex degrees of freedom along the diagonals of the three matrices.
The vacuum moduli space is therefore $\mathbb{R}^{6(N-1)}/S_N$, where we mod out by the Weyl group of $SU(N)$.
}
Consider the conifold theory~\cite{Klebanov:1998hh} whose quiver is shown in Figure~\ref{fig:conifold} and whose superpotential is
\be
W = \text{tr}\, \big( \varphi_1 \chi_1 \varphi_2 \chi_2 - \varphi_1 \chi_2 \varphi_2 \chi_1 \big) ~.
\ee
\begin{figure}[h]
    \centering
    \includegraphics[width=0.3\textwidth]{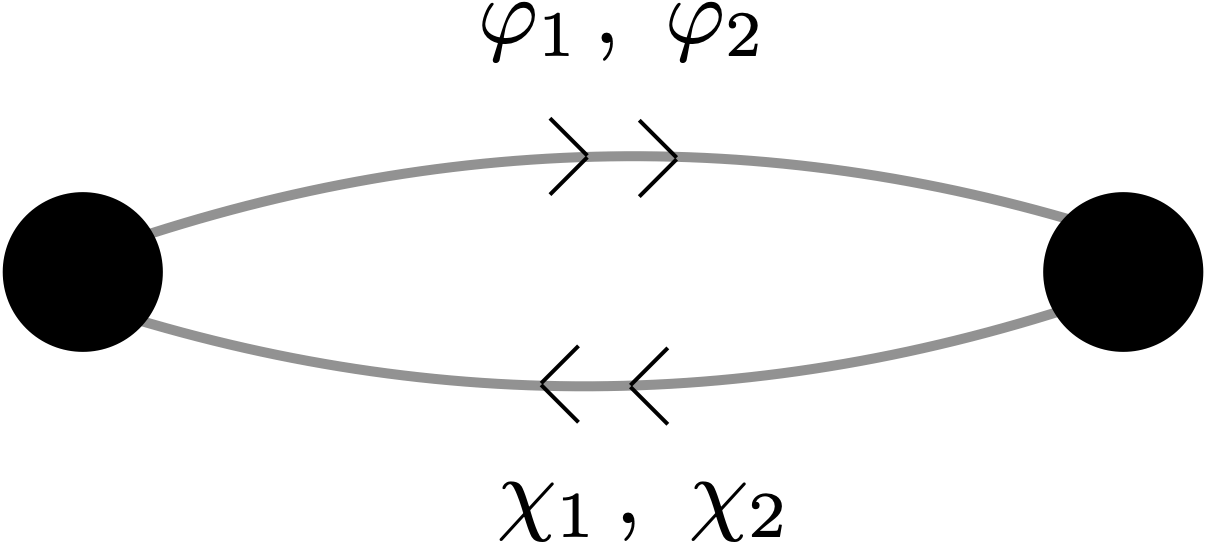}
    \caption{\textsf{The quiver of the conifold.}}
    \label{fig:conifold}
\end{figure}

\noindent
The gauge group is $SU(N)\times SU(N)$.
The matter content of the theory consists of two chiral superfields $\varphi_1$ and $\varphi_2$ that transform as $(\bm{N},\bm{\overline{N}})$ and two chiral superfields $\chi_1$ and $\chi_2$ that transform as $(\bm{\overline{N}},\bm{N})$.
The tensor product of the fundamental and antifundamental representations of $SU(N)$ gives the singlet plus the adjoint.
A GIO is constructed by contracting fundamental indices (raised) with antifundamental indices (lowered):
\be
\varphi_i \chi_j := \text{tr}\, \varphi_i \chi_j = \sum_{a=1}^N \sum_{b=1}^N \big(\varphi_i\big)^a_b \big(\chi_j)^b_a ~.
\ee
Since we take traces over the color indices, the resulting operator is a singlet --- \textit{i.e.}, it transforms in the $(\mathbf{1},\mathbf{1})$ representation of $SU(N)\times SU(N)$ --- and is thus gauge invariant.
To streamline notation we will leave the traces implicit when writing a GIO.
The full list of GIOs for the conifold theory is
\be
\{O_1, O_2, O_3, O_4\} = \{ \varphi_1 \chi_1, \varphi_1 \chi_2, \varphi_2 \chi_1, \varphi_2 \chi_2 \} ~.
\ee
For simplicity, we take $N=1$ so that the fields transform in $U(1)$.
As the fields are functions that commute, the superpotential vanishes, and there are no F-terms.
There is one relation among the GIOs:
\be
\mathcal{M} = O_1 O_4 - O_2 O_3 = 0 ~.
\ee
This is the vacuum moduli space of the theory.
It is as well the defining equation of the conifold itself.
The gauge theory on the worldvolume of a D$3$-brane at the conifold singularity is reproduced by the quiver in Figure~\ref{fig:conifold}.
This is a crucial observation: the vacuum moduli space of the $\mathcal{N}=1$ theory probes the transverse space to the brane~\cite{Douglas:1996sw}.
This applies to a host of theories, \textit{e.g.}, to branes at ADE singularities on ALE spaces.
Since $N$ parallel D3-branes at the singularity are mutually BPS, the vacuum moduli space of the $SU(N)\times SU(N)$ theory is given by $\text{Sym}^N(\mathcal{M})$, the $N$-fold symmetric product.

\subsection{Matter content and GIOs of the MSSM}\label{sec:mssm}

As we aim to apply the algorithm from Section~\ref{sec:algorithm} to determine the vacuum moduli space of the minimal supersymmetric Standard Model (MSSM), let us review the field content and the minimal basis of GIOs.
The matter fields transform in fundamental, antifundamental, and singlet representations of the product group $G=SU(3)_C\times SU(2)_L\times U(1)_Y$.
The complexification is therefore $G^c = SL(3,\mathbb{C})\times SL(2,\mathbb{C})\times \mathbb{C}^*$.
\begin{table}[h]
\centering
\begin{tabular}{|c|c|c|c|}\hline
field & representation & number & description \\ \hline \hline
$Q_{a, \alpha}^i$ & $(\bm{3},\bm{2},\frac16)$ & $18$ & doublet quarks \\ \hline
$u_a^i$ & $(\bm{\overline 3},\bm{1},-\frac23)$ & $9$ & singlet up-type quarks \\ \hline
$d_a^i$ & $(\bm{\overline 3},\bm{1},\frac13)$ & $9$ & singlet down-type quarks \\ \hline
$L_{\alpha}^i$ & $(\bm{1},\bm{2},-\frac12)$ & $6$ & doublet leptons \\ \hline
$e^i$ & $(\bm{1},\bm{1},1)$ & $3$ & singlet leptons \\ \hline
$\nu^i$ & $(\bm{1},\bm{1},0)$ & $3$ & neutrinos \\ \hline
$H_\alpha$ & $(\bm{1},\bm{2},\frac12)$ & $2$ & up-type Higgs \\ \hline
$\barH_\alpha$ & $(\bm{1},\bm{2},-\frac12)$ & $2$ & down-type Higgs \\ \hline
\end{tabular}
\caption{\textsf{Indices and field content conventions for the MSSM.
The enumeration counts each component superfield individually.}}
\label{tab:fields}
\end{table}

\noindent
To set notation, we take Latin indices from the middle of the alphabet $i,j,\ldots$ to be flavor (generation) indices that run over $1,2,3$.
Whereas the gauge indices of $G$ get contracted, multiplicities for the GIOs occur because the flavor indices can remain.
Latin indices from the beginning of the alphabet $a,b,\ldots$ also run over $1,2,3$; these label the color $SU(3)_C$.
Greek indices $\alpha,\beta,\ldots$ run over $1,2$ and label $SU(2)_L$.
The $U(1)_Y$ hypercharge assignments are given treating all the fields as being left handed.
The chiral superfields in the MSSM contain the fermions found in the Standard Model.
As there are fields $Q^i_{a,\alpha}$ charged under all the factors of the gauge group, unlike the example of the conifold, the MSSM is not a quiver gauge theory.\footnote{
It embeds into one, however~\cite{Aldazabal:2000sa,Berenstein:2001nk}.
}
A pair of Higgs doublets are necessary for two reasons: the superpotential must be holomorphic and we must ensure that gauge anomalies associated to triangle diagrams vanish.\footnote{
If, for instance, the Standard Model group embeds in an $SO(10)$ or $SU(5)$ grand unified theory (GUT), we may include additional pairs of Higgs doublets.
In this work, we do not consider this possibility.}
There are $52$ component fields in total that we list in Table~\ref{tab:fields}.
The minimal list of $976$ GIOs constructed from these is quoted in Table~\ref{tab:giofull} in Appendix~\ref{app:GIOs}~\cite{Gherghetta:1995dv,Gray:2006jb,He:2014oha}.

We first consider the MSSM with the superpotential
\bea\label{eq:renorm}
W_{\rm minimal} = \; C^0 \sum_{\alpha, \beta} H_\alpha \barH_\beta \eps + \sum_{i,j} C^1_{ij} \sum_{\alpha, \beta, a} Q^i_{a,\alpha} u^j_a H_\beta \eps \nonumber \\
 + \sum_{i,j} C^2_{ij} \sum_{\alpha, \beta, a} Q^i_{a,\alpha} d^j_a \barH_\beta \eps + \sum_{i,j} C^3_{ij} e^i \sum_{\alpha, \beta} L^j_{\alpha} \barH_\beta \eps ~.
\eea
This neglects neutrinos, which we add in \S\ref{sec:six}; we have been explicit about sums on gauge and flavor indices.
The $C^A$ are coefficient matrices whose elements are generic.
We do not introduce hierarchies by hand.
In principle, the coefficients in the superpotential are real numbers, but for purposes of the calculation, they could just as well be complex.
A computation shows that if the $C^A$ are invertible, we may take them to be identity matrices.
The superpotential includes the $\mu$-term for the Higgs doublets and Yukawa couplings for the quarks and leptons leading to Dirac masses for particles in the Standard Model Lagrangian upon electroweak symmetry breaking.
All of the terms in~\eref{eq:renorm} are renormalizable.
On phenomenological grounds, we exclude the renormalizable GIOs $LH$, $LLe$, and $QdL$ that violate R-parity from the superpotential~\eref{eq:renorm}.

Because they are GIOs in their own right, we can write different terms involving neutrinos.
The MSSM superpotential we consider becomes
\bea
W_\text{MSSM} & = & W_\text{minimal} + W_\text{neutrinos} ~, \\
W_\text{neutrinos} & = & \sum_{i,j}C^4_{ij} \nu^i\nu^j+\sum_{i,j} C^5_{ij} \nu^i \sum_{\alpha,\beta} L^j_{\alpha} H_\beta \eps ~. \label{eq:neut}
\eea
The two terms in~\eref{eq:neut} give rise to Majorana and Dirac masses for the neutrino.
We do not include linear terms (tadpoles), which can be absorbed via field redefinition, or R-parity violating $\nu^3$ terms.
Adopting a conservative view, we assume that there are exactly three species of active neutrinos.
That is to say, consistent with experiment, we do not include any sterile neutrinos.

\subsection{History and previous work on the problem}
The study of vacuum moduli spaces of supersymmetric theories has a significant pedigree~\cite{Witten:1981nf,Alvarez-Gaume:1981exv,Seiberg:1994bz,Seiberg:1994rs,Intriligator:1995au,Hanany:1996ie}.
The possibility of using the matter content and GIOs to solve for the vacuum moduli space of the MSSM was already considered as early as 1995.
The authors of~\cite{Gherghetta:1995dv} write:
{\small
\begin{quotation}
\noindent\textit{``There do exist mathematical algorithms well-known in algebraic geometry (\textit{e.g.}, the methods of resultant polynomials or Gr\"obner bases) which can in principle be used to decide whether such systems of simultaneous non-linear equations have non-trivial solutions.
Unfortunately, these methods are of no practical use at the level of complexity encountered here... the required number of algebraic operations is demonstrably finite, but quite astronomical.''}
\end{quotation}
}

Our interest in this problem began with~\cite{Gray:2005sr,Gray:2006jb}.
These works established that the vacuum moduli spaces of subsectors of the MSSM have identifiable geometry.
For instance, restricting to the electroweak sector, where, consistent with the unbroken $SU(3)_C$ symmetry in Nature, the vevs of the squarks are set to zero \textit{ab initio}, the vacuum moduli space is an affine cone over the Veronese surface, which is $\mathbb{P}^2 \hookrightarrow \mathbb{P}^5$.
(This structure is further discussed in~\cite{He:2014loa}.)
Moreover, it was observed that these geometric features persist when we incorporate higher order non-renormalizable terms into the superpotential that preserve physically desirable properties such as R-parity, which ensures the stability of the proton and leaves a stable lightest supersymmetric particle (LSP) after supersymmetry breaking as a candidate for the dark matter that today comprises some $25\%$ of the energy density of the Universe.
Based on these examples, the authors suggested that geometric structure in the vacuum is a precept for phenomenology and a selection principle for extensions of the Standard Model Lagrangian.
The paper~\cite{He:2014oha} investigated the structure of the vacuum moduli space of the electroweak sector scanning over the number of generations of particles and the number of pairs of Higgs doublets and made the observation that only when there are three generations are the vacuum moduli spaces toric.
The followup~\cite{He:2015rzg} considered the role of R-parity in three generation models in greater detail.
If we turn to the SQCD sector of the theory alone, the vacuum moduli space of the $SU(N)$ theory is Calabi--Yau for arbitrary numbers of flavors and colors, a fact deduced from the palindromicity of the numerator of the associated Hilbert series~\cite{Gray:2008yu}.

Proceeding to the full MSSM,~\cite{Xiao:2019uhh} employs the plethystic program~\cite{Feng:2007ur} to generate Hilbert series in the case where the superpotential is set to zero.\footnote{
Standard Model flavor invariants are similarly calculated and discussed in~\cite{Hanany:2010vu}.}
In the following sections, we deduce the vacuum moduli space of the full MSSM with the superpotential~\eref{eq:renorm} using Gr\"obner basis methods.
Just as~\cite{Gherghetta:1995dv} noted, this is indeed an astronomical task, but one that is in fact tractable using Gr\"obner basis methods from computational algebraic geometry. The key is to exploit a very special algebraic aspect of the problem: a multigraded structure. 

\section{The components of the master space}\label{sec:three}


The superpotential we use is~\eref{eq:renorm}:
\bea\label{renorm}
W_{\rm minimal} = \; C^0 \sum_{\alpha, \beta} H_\alpha \barH_\beta \eps + \sum_{i,j} C^1_{ij} \sum_{\alpha, \beta, a} Q^i_{a,\alpha} u^j_a H_\beta \eps \nonumber \\
 +\sum_{i,j} C^2_{ij} \sum_{\alpha, \beta, a} Q^i_{a,\alpha} d^j_a \barH_\beta \eps + \sum_{i,j} C^3_{ij} e^i \sum_{\alpha, \beta} L^j_{\alpha} \barH_\beta \eps \; .
\eea
In the initial part of the analysis, we do not consider the singlet right handed neutrinos in~\eref{eq:neut}; we add neutrinos in \S 6.

\subsection{Changing basis via the Yukawa couplings}
A key simplification is to rewrite the superpotential using the following variables:
\begin{equation}
  \tilde{e}^i := \sum_j C^3_{ji} e^j \;, 
  \qquad \tilde{u}^i_a := \sum_j C^1_{ij} u^j_a \;, 
  \qquad \tilde{d}^i_a := \sum_j C^2_{ij} d^j_a \;.
\end{equation}
By scaling variables, we may assume that $C^0 = 1$.
We also assume as in~\cite{Gherghetta:1995dv} that these coupling matrices $C^1$, $C^2$, and $C^3$ are all invertible.
The resulting superpotential in these variables is
\bea\label{renorm}
W_{\rm minimal} = \; \sum_{\alpha, \beta} H_\alpha \barH_\beta \eps + \sum_{i} \sum_{\alpha, \beta, a} Q^i_{a,\alpha} \tilde{u}^i_a H_\beta \eps \nonumber \\
 +\sum_{i} \sum_{\alpha, \beta, a} Q^i_{a,\alpha} \tilde{d}^i_a \barH_\beta \eps + \sum_{i}  \tilde{e}^i \sum_{\alpha, \beta} L^i_{\alpha} \barH_\beta \eps \;.
\eea
The transformation $(C^1, C^2, C^3) \in GL(3, \CC) \times GL(3, \CC) \times GL(3, \CC)$ gives a linear action on $\CC^{49}$.

The next proposition shows that the resulting change of variables does indeed preserve the linear span of each of the $28$ GIO types appearing in Appendix~\ref{app:GIOs}, and therefore we may change variables to the tilde versions, and the image will be linearly equivalent to the image under the original map.

\begin{prop}
The group $GL(3, \CC) \times GL(3, \CC) \times GL(3, \CC)$ acts linearly on the span of each of the $28$ GIO types appearing in Appendix~\ref{app:GIOs}.
\end{prop}

\begin{proof}
There are $28$ cases, which we break into three classes:
\[
\begin{array}{ccl}
\mbox{Case 1:}& &\quad LH \,,~ H\barH \,,~ QQQL \,,~ QQQ\barH \,,~ (QQQ)_4LLH \,,~ (QQQ)_4LH\barH \,,~ (QQQ)_4H\barH\barH \;. \\
\mbox{Case 2:}& &\quad LLe \,,~ L\barH e \,,~ (QQQ)_4LLLe \,,~ (QQQ)_4LL\barH e \,,~ (QQQ)_4L\barH\barH e \,,~ (QQQ)_4\barH\barH\barH e \;.\\
\mbox{Case 3:}& &\quad udd \,,~ QdL \,,~ QuH \,,~
    Qd\barH \,,~ QuQd \,,~ QuLe \,,~ uude \,,~ Qu\barH e \,,~ dddLL \,,~ uuuee \,, \\
    & &\quad QuQue \,,~ QQQQu \,,~ dddL\barH \,,~ uudQdH \,,~ uudQdQd \;.
    \end{array}
    \]
The seven GIO types in Case~1 have no $u, d$ or $e$ in them, hence each individual GIO in each of these types is invariant under this group action.

For Case~2, neither $u$ nor $d$ appear, and since each GIO in these types has a linear factor of $e_i$, for some $i$, each type is visibly invariant under the group action.

For Case~3, we need only check that the action of $C^1$ on $u$, $C^2$ on $d$ and $C^3$ on $e$ leave the linear span of each of the $15$ GIO types invariant.
We check this computationally as follows.
First, we work on each $GL(3,\mathbb{C})$ separately.
Second, we create a base field $K$ which is the fraction field of a polynomial ring whose nine variables are the entries in a $3 \times 3$ matrix.
Then, we compute the linear span of the generators of the ideal consisting of all the GIOs of a given type, where we have acted on the respective $u$, $d$, $e$ by this matrix.
We compare this linear span (over $K$) with the original span and verify that they agree.
\end{proof}

\begin{exm}\rm
The above computation is particularly transparent for the $dddLL$ type,
because the action on the $d$ variables is essentially multiplication of the $3 \times 3$
matrix $D = (d^i_a)$ by the matrix $C^2 = (C^2_{ij})$, and since each $dddLL$ is a product of the determinant of $D$ with a polynomial in the $L$ variables, the action 
multiplies each polynomial by the determinant of the matrix $C^2$, and hence the linear span of the set of $dddLL$ GIOs is invariant.
\end{exm}

\subsection{Equations for the master space}
The master space is the zero set of the ideal $J_W$ generated by the $49$ partial derivatives of $W_{minimal}$.

\begin{exm}\rm
Notice that the variables $\tu^i_a$ only appear in one of the four summations in the superpotential.
Thanks to  the change of variables, the partial derivatives with respect to the $\tu$ variables have a particularly simple form:
\[\left(\!\begin{array}{c}
-H_{2}Q_{1,1,1}+H_{1}Q_{1,1,2}\\
-H_{2}Q_{1,2,1}+H_{1}Q_{1,2,2}\\
-H_{2}Q_{1,3,1}+H_{1}Q_{1,3,2}\\
-H_{2}Q_{2,1,1}+H_{1}Q_{2,1,2}\\
-H_{2}Q_{2,2,1}+H_{1}Q_{2,2,2}\\
-H_{2}Q_{2,3,1}+H_{1}Q_{2,3,2}\\
-H_{2}Q_{3,1,1}+H_{1}Q_{3,1,2}\\
-H_{2}Q_{3,2,1}+H_{1}Q_{3,2,2}\\
-H_{2}Q_{3,3,1}+H_{1}Q_{3,3,2}
\end{array}\!\right) \;.
\]
These equations are the $2 \times 2$ minors involving the first column of the $2 \times 10$ matrix
\[\left(\!\begin{array}{cccccccccc}
H_{1}&Q_{1,1,1}&Q_{1,2,1}&Q_{1,3,1}&Q_{2,1,1}&Q_{2,2,1}&Q_{2,3,1}&Q_{3,1,1}&Q_{3,2,1}&Q_{3,3,1}\\
H_{2}&Q_{1,1,2}&Q_{1,2,2}&Q_{1,3,2}&Q_{2,1,2}&Q_{2,2,2}&Q_{2,3,2}&Q_{3,1,2}&Q_{3,2,2}&Q_{3,3,2}
\end{array}\!\right) \;.
\]
Notice that the partial derivatives with respect to the $\td^i_a$, $\te^i$, and $L^i_\alpha$ variables behave similarly.
It turns out that many of the $49$ equations arise as in this previous example.
\end{exm}
Consider the following matrices
$M_Q, M_{du}, M_H, M_L$:
\[
\begin{array}{{ccl}}
    \vspace{.05in} M_Q & := & \left(\!\begin{array}{ccccccccc}
Q_{1,1,1}&Q_{1,2,1}&Q_{1,3,1}&Q_{2,1,1}&Q_{2,2,1}&Q_{2,3,1}&Q_{3,1,1}&Q_{3,2,1}&Q_{3,3,1} \;,\\
Q_{1,1,2}&Q_{1,2,2}&Q_{1,3,2}&Q_{2,1,2}&Q_{2,2,2}&Q_{2,3,2}&Q_{3,1,2}&Q_{3,2,2}&Q_{3,3,2}
\end{array}\!\right) \;,\\
   \vspace{.05in}  M_{du}& := &
\left(\!\begin{array}{ccccccccc}
\tu_{1,1}&\tu_{1,2}&\tu_{1,3}&\tu_{2,1}&\tu_{2,2}&\tu_{2,3}&\tu_{3,1}&\tu_{3,2}&\tu_{3,3}\\
-\td_{1,1}&-\td_{1,2}&-\td_{1,3}&-\td_{2,1}&-\td_{2,2}&-\td_{2,3}&-\td_{3,1}&-\td_{3,2}&-\td_{3,3}
\end{array}\!\right) \;,\\
      \vspace{.05in} M_H & := & \left(\!\begin{array}{cccc}
H_1 & \barH_1 & H_1& H_2 \\ 
H_2 & \barH_2 & \barH_1 & \barH_2
\end{array}\!\right) \;,\\
M_L & := & \left(\!\begin{array}{ccc}
L^1_1 & L^2_1 & L^3_1 \\ 
L^1_2 & L^2_2 & L^3_2
\end{array}\!\right) \;.
\end{array}
\]

\begin{comp}
The ideal $J_{W_{minimal}}$ defining the master space is generated by:
\begin{itemize}[leftmargin=.5in,itemsep=0in]
\item The six monomials 
$\{ \barH_1 \te^1,\barH_1 \te^2,\barH_1 \te^3,\barH_2 \te^1,\barH_2 \te^2,\barH_2 \te^3 \} \;.$
\item The four forms with a linear term in $H_\alpha$ or $\bar{H}_\alpha$: for each $\alpha \in \{1,2\}$, we have
\[ 
\begin{array}{ccl}
  \barF_\alpha & = & \barH_\alpha - \sum\limits_{i,a} \tu^i_a Q^i_{a,\alpha} \;,\\
  F_\alpha & = & H_\alpha + \sum\limits_{i,a} \td^i_a Q^i_{a,\alpha} + \sum\limits_{i} L^i_\alpha \te^i \;.
  \end{array}
\]
\item Let $(M_H)_{\{i_1,i_2, \ldots i_k\}}$ denote columns $\{i_1,i_2,\ldots i_k\}$ of $M_H$.
The remaining $39$ generators for $J_W$ come from the $2 \times 2$ minors of the following five matrices, but only those minors which involve the first column of each matrix. We use a vertical bar to denote concatenation.

\[
\begin{array}{ccc}
\vspace{.05in}N_1 &=& \Big(M_H\Big)_{1} \ \Big| \ M_Q \;, \\ 
\vspace{.05in}N_2 &=& \Big(M_H\Big)_2 \  \Big| \ M_Q \;, \\
\vspace{.05in}N_3 &=& \Big(M_H\Big)_3 \ \; \Big| \ M_{du} \;, \\
\vspace{.05in}N_4 &=& \Big(M_H\Big)_4 \ \ \Big| \ M_{du} \;, \\
\vspace{.05in}N_5 &=& \Big(M_H\Big)_2 \  \Big| \ M_L \;.
\end{array}
\]
\comment{
    \item 
    \[N_1 = \left(\!\begin{array}{cccccccccc}
H_{1}&Q_{1,1,1}&Q_{1,2,1}&Q_{1,3,1}&Q_{2,1,1}&Q_{2,2,1}&Q_{2,3,1}&Q_{3,1,1}&Q_{3,2,1}&Q_{3,3,1}\\
H_{2}&Q_{1,1,2}&Q_{1,2,2}&Q_{1,3,2}&Q_{2,1,2}&Q_{2,2,2}&Q_{2,3,2}&Q_{3,1,2}&Q_{3,2,2}&Q_{3,3,2}
\end{array}\!\right)
\]
    \item 
    \[N_2 = \left(\!\begin{array}{cccccccccc}
\barH_{1}&Q_{1,1,1}&Q_{1,2,1}&Q_{1,3,1}&Q_{2,1,1}&Q_{2,2,1}&Q_{2,3,1}&Q_{3,1,1}&Q_{3,2,1}&Q_{3,3,1}\\
\barH_{2}&Q_{1,1,2}&Q_{1,2,2}&Q_{1,3,2}&Q_{2,1,2}&Q_{2,2,2}&Q_{2,3,2}&Q_{3,1,2}&Q_{3,2,2}&Q_{3,3,2}
\end{array}\!\right)
\]
     \item 
     \[N_3 = 
\left(\!\begin{array}{cccccccccc}
\barH_{1}&u_{1,1}&u_{1,2}&u_{1,3}&u_{2,1}&u_{2,2}&u_{2,3}&u_{3,1}&u_{3,2}&u_{3,3}\\
H_{1}&-d_{1,1}&-d_{1,2}&-d_{1,3}&-d_{2,1}&-d_{2,2}&-d_{2,3}&-d_{3,1}&-d_{3,2}&-d_{3,3}
\end{array}\!\right)
\]
    \item 
         \[N_4 = 
\left(\!\begin{array}{cccccccccc}
\barH_{2}&u_{1,1}&u_{1,2}&u_{1,3}&u_{2,1}&u_{2,2}&u_{2,3}&u_{3,1}&u_{3,2}&u_{3,3}\\
H_{2}&-d_{1,1}&-d_{1,2}&-d_{1,3}&-d_{2,1}&-d_{2,2}&-d_{2,3}&-d_{3,1}&-d_{3,2}&-d_{3,3}
\end{array}\!\right)
\]
 \item 
      \[N_5 = \left(\!\begin{array}{cccc}
\barH_{1} & L^1_1 & L^2_1 & L^3_1 \\ 
\barH_{2} & L^1_2 & L^2_2 & L^3_2
\end{array}\!\right)
\]
}
\end{itemize}
\end{comp}
The first ingredient to understanding the vacuum moduli space is an analysis of the master space, which is provided by the next theorem.

\begin{thm}\label{IdealMasterSpaceComponents}
The master space has three components $X_1$, $X_2$, and $X_3$, which
\begin{itemize}[leftmargin=.5in,itemsep=0in]
\item are reduced and irreducible,
\item are of dimensions $23$, $27$, $41$, respectively,
\item are defined by prime ideals $I_i$ described explicitly below.
\end{itemize}
\end{thm}

\begin{proof}
Define ideals $I_1, I_2, I_3$ as below:
\[
\begin{array}{ccl}
\vspace{.2in} I_1 & = & I_{2 \times 2}(M_1) + I_{2 \times 2}(M_2) + \langle \te^1, \te^2, \te^3 \rangle + \langle F_1, F_2, \barF_1, \barF_2 \rangle \,, \mbox{ where}\\
\vspace{.2in} & & M_1 = \Big(M_H\Big)_{1,2} \ \Big| \ M_Q  \ \Big| \ M_L  \,, \mbox{ and}\\ 
 \vspace{.2in}& &M_2 = \Big(M_H\Big)_{3,4} \  \Big| \ M_{du} \,.\\
\vspace{.25in} I_2 & = & I_{2 \times 2}(N_1)  + 
\langle \barH_1, \barH_2 \rangle +
\langle \tu^i_a \rangle_{i,a} + 
\langle F_1, F_2 \rangle \,.\\
I_3 & = & 
\langle H_1, H_2, \barH_1, \barH_2 \rangle +
\langle F_1, F_2, \barF_1, \barF_2 \rangle \,.
\end{array}
\]

A direct computation shows that 
\[ 
J_W = I_1 \cap I_2 \cap I_3 \,,
\]
and that the dimensions are as described. It is possible to prove that the $I_i$ are prime ideals directly.
For example, for the ideal $I_2$, we quotient by the $\barH$ and $\tilde{u}$ variables. Vanishing of the $F_i$ allows us to replace the $H_i$ in the first column of $N_1$ with two bihomogeneous polynomials of type $dQ+Le$. Now simplify the resulting matrix using the $M_Q$ columns to further reduce $N_1$, to the matrix
    \[
\Big( \sum\limits_{i} L^i_1 \te^i, \sum\limits_{i} L^i_2 \te^i \Big)^T \Big| M_Q \,.
      \]
As the first column consists of two independent forms in variables distinct from the $Q$ variables, the resulting ideal of $2 \times 2$ minors is prime. 

Primality may also be verified computationally; this completes our analysis of the components of the master space.
\end{proof}
\comment{
\item For $I_3$, quotienting out the $H$ and $\barH$ variables reduces the remaining four generators to
\[ 
\begin{array}{l}
 \sum\limits_{i,a} \tu^i_a Q^i_{a,\alpha} \\
 \sum\limits_{i,a} \td^i_a Q^i_{a,\alpha} + \sum\limits_{i} L^i_\alpha \te^i, \mbox{ where }\alpha \in \{ 1,2 \}.
  \end{array}
\]
Note that for the bottom set of two equations, each contains a variable not arising in any of the other equations (the $L^i_1, L^i_2$).
\end{itemize}

To prove that the dimensions are as claimed, we need the fact~\cite{Eisenbud} that an ideal generated by the $2 \times 2$ minors of a generic $2 \times m$ matrix has codimension $m-1$. With this in hand, we compute: \newline
$\bullet$ For $I_1$, \newline
$\bullet$ For $I_2$, quotienting by the $\tilde{u}$ and $\barH$ drops dimension by $9+2=11$, and killing the $H$ drops dimension by two more. The $2\times 2$ minors of $I_{2 \times 2}(N_1)$ drop the dimension by another $9$ from our remarks above (note that the variables involved are distinct, so codimension behaves additively). Therefore $codim(I_2) = 9+2+2+9 =22$ and the dimension of $V(I_2)$ as an affine variety is $27$. \newline
$\bullet$ For $I_3$, the $H$ and $\barH$ variables yield codimension of four, and the remaining four equations have enough independent variables to increase the codimension four more. Hence $V(I_3)$ is of dimension $41$, and in fact is a complete intersection.}

\section{The components $M_1$ and $M_2$ of the vacuum moduli space}\label{sec:four}

We first analyze the image of the  components $X_1$ and $X_2$ of the master space.
It turns out on these two components that there are relatively few non-zero GIO types.
In contrast, for $X_3$ the majority of the GIO types are non-zero, making it far more difficult to analyze, so we defer that analysis to \S\ref{sec:five}. 

\subsection{The image $M_1$ of component $X_1$}

Recall from the previous section that the ideal of $X_1$ is defined by
\[ I_{2 \times 2}(M_1) + I_{2 \times 2}(M_2) + \langle \te^1, \te^2, \te^3 \rangle + \langle F_1, F_2, \barF_1, \barF_2 \rangle \,.
\]

\begin{prop}
The image $M_1$ under all $973$ GIOs is a line.
More specifically
\begin{enumerate}[leftmargin=.5in]
\item
Every GIO type  except $udd$ vanishes on the zero set of
\[ I_{2 \times 2}(M_1)  + \langle \te^1, \te^2, \te^3 \rangle \,,
\]
and consequently on the component $X_1$.
\item The image of $X_1$ under the $udd$ type has image a line in $\CC^9 \subset \CC^{973}$.
In fact, since six of the nine $udd$ GIOs vanish on $X_1$, the line lies in a $\CC^3$.
\end{enumerate}
\end{prop}

\begin{proof}
Part (1) is shown by a \textit{Macaulay2} computation, and that all but three of the nine $udd$s on $X_1$ vanish.
The same computation shows that the remaining three $udd$s are all scalar multiples of each other, so define a line in the image space.
\end{proof}

\subsection{The image $M_2$ of component $X_2$}
The analysis of $M_2$ is more complicated than that of $X_1$, and requires new methods. Recall that by Theorem~\ref{IdealMasterSpaceComponents}, the ideal $I_2$ of $X_2$ is defined by
\[\begin{array}{ccc}
\vspace{.05in} I_2 & = & I_{2 \times 2}(N_1)  + 
\langle \barH_1, \barH_2 \rangle +
\langle \tu^i_a \rangle_{i,a} + 
\langle F_1, F_2, \barF_1, \barF_2 \rangle \,.\\
\end{array}
\]
By reducing the $28$ GIO types with respect to the ideal $I_2$, we find that the only nonzero GIO types on $X_2$ are
\[LH \,,~ LLe \,,~ QdL \,,~ dddLL \,,~ (QQQ)_4LLLe \,.\]
Furthermore, since $\barH_\alpha$ vanishes on $X_2$, as do the $\tu^i_a$, the three LH elements may be written in terms of the $QdL$ and $LLe$ GIOs, so up to a change of coordinates, these three elements are not needed in the map to $\CC^{973}$.
This leaves $69$ non-vanishing GIOs.

\begin{thm}
    The image $M_2$ of $X_2$ under the map $\phi$ has dimension $15$, is rational, and sits inside a linear space $\CC^{69} \subset \CC^{973}$.
    Furthermore, we compute explicit generators for the ideal of the image in $\CC^{69} \subset \CC^{973}$.
\end{thm}
The diagram below is a roadmap of the high level steps in our analysis, showing the maps and describing the spaces that will be the key players. 
\begin{center}
    \begin{tikzcd}
      & \CC^{49} \arrow[rrr, "\phi"] & & & \CC^{973} \\
      & \CC^{38} \arrow[u, hook]
      &
      & &  \CC^{69} \arrow[u, hook] \\
      & X_2 \arrow[rrr, bend left, "\phi"] \arrow[u, hook] & & & M_2 \arrow[u, hook] \\
      \CC^{29} \supset U_2 \arrow[ur, hook, "f"] \arrow[urrrr, dotted] \arrow[dr,  rightarrow, "g"] \\
      & \CC^{20}  \arrow[rrruu, twoheadrightarrow, "h"] & &
    \end{tikzcd}
\end{center}

\begin{itemize}
\item The $\CC^{38}$ is a reflection of the fact that the $11$ variables $\tu^i_a$ and $\barH_\alpha$ vanish on $X_2$.
\item The subset $U_2 \subset X_2$ is the open affine subset where $H_1 \ne 0$.
On this open set, setting $t = H_2/H_1$, then $H_2 = t\, H_1$. 
Since the ideal of $X_2$ contains $I_{2 \times 2}(N_1)$, on $U_2$ we also have $Q_{i,j,2} = t\, Q_{i,j,1}$, for all $i,j$.
So $U_2$ is (isomorphic to) a subvariety of $\CC^{29}$, with the $29$ variables 
\[ t \,,\ Q_{i,j,1} \,\ L^i_{\alpha} \,\ \te^i \,\ H_\alpha \,\ \td^i_a \,. \]
In particular, the $Q_{i,j,2}$ become redundant in describing $U_2$.
\item The equations of $U_2$ in $\CC^{29}$ are:
\[\begin{array}{rcl}
  F_1 & = & H_1 + \sum\limits_{i,a} \td^i_a Q^i_{a,1} + \sum\limits_{i} L^i_1 \te^i \,, \\
  F_2 - tF_1 & = & 
  \sum\limits_i (t\, L^i_1 - L^i_2)\, \te^i \,.
\end{array}\]
This visibly defines $U_2$ as a complete intersection in $\CC^{29}$, or, by eliminating $H_1$, as a hypersurface in $\CC^{28}$.
\item The next step is to factor the GIO map from $U_2$ to $M_2$ through $\CC^{20}$ in the diagram.
The $\CC^{20}$ arises from the fact that the GIOs factor on $U_2$ in a nice way, in terms of the following polynomials:
\begin{itemize}
    \item The three $\te^i$'s.
    \item The three $2 \times 2$ minors $\Delta_i$ of the $L^i_\alpha$ matrix.
    \item The determinant $\det D$ of the $3 \times 3$ matrix $D = \td^i_a$.
    \item The determinant $\det Q$ of the $3 \times 3$ matrix $Q = Q_{i,a,1}$.
    \item The nine entries $DQ_{i,j}$ of the product $DQ$.
    \item The three elements $tL_i := tL^i_1 - L^i_2$.
\end{itemize}
There are a total of twenty elements in the list above; these polynomials give a map $g$ from the hypersurface $U_2 \subset \CC^{28}$ to $\CC^{20}$.
\item Since every GIO that is nonvanishing on $U_2$ factors into a product of elements in the list above, we obtain an induced map $h \colon \CC^{20} \to \CC^{69}$, such that the composition $hg$ is equal to $\phi f$.
Because of the factorization of the GIOs, the entries of $h$ are monomials.
The $69$ GIOs factor (up to a constant scaling factor) as products of the $20$ elements above as:
\begin{itemize}
    \item ($LLe$, $9$ of these): $\Delta_i \cdot \te^j$.
    \item ($QdL$, $27$ of these): $DQ_{i,j}\cdot  tL_i$.
    \item ($dddLL$, $3$ of these): $(\det D) \cdot \Delta_i $.
    \item ($(QQQ)_4LLLe$, $30$ of these): $(\det Q) \cdot \te^i \cdot m$, for the ten
    degree three monomials $m$ in the $tL_i$.
\end{itemize}
\item The image of $U_2$ under $g$
in $\CC^{20}$ is defined by the following three equations.
\[\begin{array}{rcl}
  \begin{pmatrix}tL_1 & tL_2 & tL_3\end{pmatrix}
  \begin{pmatrix}
  \te^1 & \Delta_1 \\
  \te^2 & \Delta_2 \\
  \te^3 & \Delta_3
  \end{pmatrix} = 0 \,, \\
  \det(DQ) - \det D \det Q = 0 \,.
\end{array}
\]
\end{itemize}
\begin{comp}\label{IdealofM2}
 The ideal of $M_2$ is defined by $6$ linear forms, $816$ quadratics, and $90$ quartics.
 While this gives an explicit list of generators for the ideal of $M_2$, the geometric structure of $M_2$ is more clearly understood via a birational projection 
\[
M_2 \stackrel{\pi}{\longrightarrow} \cV_2 \subseteq \CC^{39} \,.
\]
We tackle this in the next subsection.
\end{comp}
\subsection{A birational model of $M_2$ connecting to classical geometry}
As we have seen, on $X_2$, the $LH$ GIOs can be written in terms of the $LLe$ and $QdL$ GIOs.
A similar result holds for the $(QQQ)_4LLLe$ GIOs on $M_2$, except they are rational functions, not polynomials.
For notational convenience, and consistency with the computations in \S 5, we denote the GIO types in use as 
\begin{itemize}
\item The $30$ $(QQQ)_4LLLe$s as $q_1, \ldots, q_{30}$.
\item The $27$ $QdL$s as $a_1, \ldots, a_{27}$.
\item The $9$ $LLe$s as $b_1, \ldots, b_9$.
\item The $3$ $dddLL$s as $i_1, \ldots, i_3$.
\end{itemize}

\begin{comp}
On $X_2$, the $(QQQ)_4LLLe$ GIOs are rational functions of the $LLe$, $QdL$, and $dddLL$ functions.
We have relations for all $j = 1, \ldots, 30$ and all $\ell = 1,2,3$,
 \[q_j = \frac{b_k C(a_*)}{i_\ell}\,, \]
 where $C(a_*)$ is a determinantal cubic in the $a$ variables,
 and $k$ depends on $\ell$ and $j$.
\end{comp}

By the above computation, to understand the component $M_2$ up to birational equivalence, it suffices to consider the image of $X_2$ under the map to $\CC^{39}$ given by the GIOs of the three types $LLe$, $QdL$, and $dddLL$.
This corresponds to the projection map $\pi \colon M_2 \subset \CC^{69} \to \CC^{39}$.
Let $\cV_2 := \pi(M_2)$ be this projection.

The remainder of this section is devoted to a proof that the geometry of $\cV_2$ may be expressed as an incidence correspondence between Segre varieties, yielding a geometric picture of the structure of $M_2$.
The following matrices play a key role in our analysis. 
\begin{itemize}
\item $M_a := \left(\!\begin{array}{ccccccccc}
a_{1}&a_{4}&a_{7}&a_{10}&a_{13}&a_{16}&a_{19}&a_{22}&a_{25}\\
a_{2}&a_{5}&a_{8}&a_{11}&a_{14}&a_{17}&a_{20}&a_{23}&a_{26}\\
a_{3}&a_{6}&a_{9}&a_{12}&a_{15}&a_{18}&a_{21}&a_{24}&a_{27}
\end{array}\!\right) \,,$
\item $M_b := \left(\!\begin{array}{ccc}
b_{1}&b_{2}&b_{3}\\
b_{4}&b_{5}&b_{6}\\
b_{7}&b_{8}&b_{9}
\end{array}\!\right) \,,$
\item $M_z := \begin{pmatrix} i_1 & i_2 & i_3 \end{pmatrix}$\, \mbox{(we use $M_z$ rather than $M_i$ as the latter looks like an index)},
\item $J := \left(\!\begin{array}{ccc}
0&0&1\\
0&-1&0\\
1&0&0
\end{array}\!\right) \,.$
\end{itemize}
\begin{comp}\label{ProjM2}
    The ideal of the $\cV_2$ is minimally generated by 189 homogeneous quadratic equations, which may be partitioned into the five sets below:
\begin{enumerate}[leftmargin=.5in]
 \item The $108$ $2 \times 2$ minors of $M_a$.
 \item The $18$ $2 \times 2$ minors of $M_b^T | M_z^T$.
 \item The $9$ entries of $M_z \cdot J \cdot M_a$.
 \item The $27$ entries of $M_b^T \cdot M_a$.
 \item The $27$ entries of $M_b \cdot J \cdot M_a$.
\end{enumerate}
  \end{comp}

\begin{lem}~\label{lemma-segre}
    The first ideal above defines the Segre variety $\Sigma_{2,8}$ of $\PP^2 \times \PP^8 \subseteq \PP^{26}$, and the $2 \times 2$ minors of $M_b^T$ define the Segre variety $\Sigma_{2,2} \subseteq \PP^8$.
\end{lem}
\begin{proof}
  Let ${\bf a} = (a_0, \ldots, a_n)$ be coordinates on $\PP^n$ and similarly ${\bf b} = (b_0, \ldots, b_m)$ be coordinates on $\PP^m$.
  We can represent elements of $\PP^{(n+1)(m+1)-1}$ by $(n+1) \times (m+1)$ matrices.
  The Segre variety $\Sigma_{n,m}$ is defined by the image of the map
  \[ ({\bf a}, {\bf b}) \mapsto \mbox{ the rank one matrix } {\bf a}^T {\bf b} \in \PP^{(n+1)(m+1)-1}. \]
The image consists of rank one matrices, and its ideal is generated by the $2 \times 2$ minors of a $(n+1) \times (m+1)$ matrix of coordinate variables on the image space.
So, in the case at hand, $\Sigma_{2,8}$ is defined by the ideal of $2 \times 2$ minors of $M_a$, and $\Sigma_{2,2}$ is defined by the ideal of $2 \times 2$ minors of $M_b$.
For more details, see Harris~\cite{Harris:1995} or Eisenbud~\cite{Eisenbud:1995}.
\end{proof}
We close this section with the promised geometric interpretation of the ideal appearing in Computation~\ref{ProjM2}:

\begin{thm}
    The affine variety $\cV_2$ is the cone over the incidence correspondence
    \[Y \subset \Sigma_{2,2} \times \Sigma_{2,8} \times \PP^2 \,, \]
    defined by the remaining equations in the above computation that are not the equations of the Segre varieties.
    It consists of all points (thought of as row vectors)
    \[ ((\theta_1,\theta_2),(\theta_3,\theta_4),\theta_5) \in (\PP^2 \times \PP^2) \times (\PP^2 \times \PP^8) \times \PP^2 \]
    such that $\theta_2 = \theta_5$ and $\theta_1 \perp \theta_3$ and $\theta_2 J \perp \theta_3$.
    These conditions define an ideal of five homogeneous quadratic polynomials of codimension four.
    Therefore, $Y$ is of dimension $12$ and is rational and smooth, and the cone $\cV_2$ is dimension $15$, rational and smooth away from the two-dimensional vertex.

\end{thm}
\begin{proof}
    By Lemma~\ref{lemma-segre}, $Y$ is a subset of the first displayed equation in the statement of the theorem.
    The $2 \times 2$ minors of $M_b^T | M_z^T$ involving the last column (which is $\theta_5^T$) force $\theta_5$ to be proportional to any row of $M_b$ ($M_b$ has rank one), accounting for the $\theta_2 = \theta_5$ equation.
    We have now accounted for all equations of types (1) and (2) above.

    For the equations of type (3), the matrix $M_a = \theta_3^T \cdot \theta_4$, 
    therefore $M_z \cdot J \cdot M_a$ yields the single equation 
     \[ {\bf \theta}_5 \cdot J \cdot {\bf \theta}_3^T = 0 \,, \]
    because $\theta_4 \ne 0$ and $ {\bf \theta}_5 \cdot J \cdot {\bf \theta}_3^T \cdot \theta_4 = 0.$
    
    For equations of type (4), we have
    \[ M_b^T \cdot M_a = \theta_2^T \cdot \theta_1 \cdot \theta_3^T \cdot \theta_4 = 0 \,, \]
    and therefore, reasoning as above,
    \[ \theta_1 \theta_3^T = 0 \,. \]

   For equations of type (5), we have 
       \[ M_b \cdot J \cdot M_a = \theta_1^T \cdot \theta_2 \cdot J \cdot \theta_3^T \cdot \theta_4 = 0 \,, \]
    and therefore, reasoning as above,
    \[ \theta_2 \cdot J \cdot \theta_3^T = 0 \,.\]
    But $\theta_2 = \theta_5$, and this duplicates the equation of type (3) above.

    $\theta_2 = \theta_5$ gives three homogeneous quadrics, and the two equations of type (3) and (4) each yield a homogeneous quadric, resulting in five quadrics.
    A computation confirms that the ideal generated by these five quadrics has codimension four.
\end{proof}

 \begin{comp}
 Restricting $M_2$ by setting $Q=u=d=0$ yields an irreducible complex affine variety of dimension five: a cone over the Segre variety $\PP^2 \times \PP^2$.
 This is precisely the variety which \S 3.3 of~\cite{He:2014oha} shows defines the vacuum moduli space restricted to the electroweak sector.  
 \end{comp}
 
\section{Analysis of the $M_3$ component of the vacuum moduli space}\label{sec:five}

Our goal is to understand the geometry of the image $M_3$ of component $X_3$ under the GIO map $\phi$.
In contrast to the analysis of $M_1$ and $M_2$, 
there are $783$ GIOs which do not vanish on $X_3$, in the following $13$ types 
\[
\begin{array}{l}
  QdL \,,~ udd \,,~ uude \,,~ uuuee \,,~ dddLL \,,~ LLe \,,~ QQQL \,,~ QuQd \,,\\
 \vspace{.2in} QuLe \,,~ QQQQu \,,~ QuQue \,,~ QQQ_4LLLe \,,~ uudQdQd \,.
  \end{array}
\]
Such a large number of GIOs makes the task of finding the image seem problematic, in particular,
the ideal defining this image in $\CC^{783}$ 
has so many generators that an explicit description
is unenlightening.

\subsection{Overview of the strategy}
Instead of finding the defining equations for the image in $\CC^{783}$, our plan is to find a projection of $M_3$
to a smaller dimensional base space $Y$, which preserves certain essential features of $M_3$, such as dimension and such that the resulting projection map from $M_3$ to $Y$ is generically (\textit{i.e.}, on an open dense subset of $M_3$) one-to-one.

We use a fast probabilistic algorithm to determine the dimension of $M_3$, which a computation shows to be $29$, as expected for the quotient by the group action.

To identify the most economical (\textit{i.e.}, in a target space
where we can understand the image) projection, we systematically search through all small subsets (\textit{e.g.}, two or three or four sets of the $13$ GIO types that are non-vanishing on $X_3$) to pinpoint those
having the image of $X_3$ of the expected dimension $29$.

\comment{
by choosing a small set of the above $13$ GIO types.
Then, we check the dimension of the image of $X_3$ under this set of GIOs, using the same fast probabilistic
algorithm.  Our goal is to find a small set of GIOs
having image of expected dimension $29$.
}

It turns out to be the case that there no subsets of pairs of GIO types having image dimension $29$, but there are $81$ subsets of three GIO types having image dimension $29$.

Our next aim is to identify one of these sets where we can prove that the projection map from $M_3$ to $Y$ is generically one-to-one. 
\begin{comp}
    There are no sets involving only one or two GIO types such that the image of $X_3$ under the maps determined by these types has image of dimension $29$. There are $81$ triples of GIO types with $29$ dimensional image, and we list several of the smallest sets in Table~\ref{tbl:two}:
   \begin{table}[h!!!]
\begin{center}{\small
\begin{tabular}{|c|c|}\hline
Types & Total number of GIOs \\
\hline \hline
$\{QdL, udd, uuuee\}$ & $42$ \\ \hline
$\{QdL, uude, dddLL\}$ & $57$ \\ \hline
$\{QdL, uude, uuuee\}$ & $60$\\ \hline
$\{QdL, udd, uude\}$ & $63$ \\ \hline
\vdots & \vdots \\
\hline
\end{tabular}}
\end{center}{\caption{\label{tbl:two}
{\textsf{Types of GIO triples with $29$ dimensional image on component $X_3$.}}
}}
\vspace{-.1in}
\end{table}
\end{comp}
    In terms of the number of polynomials it contains, the smallest GIO set above is $\{QdL, udd, uuuee\}$.
    However, a computation shows that while it has a $29$-dimensional image, the map from $M_3$ to it is generically two-to-one, so this set is not suited to our purposes.
    It turns out that the set of GIO types $\{QdL, udd, uude\}$ yields a manageable set of equations, and we can prove that the projection is generically one-to-one.
    From a physics point of view, this set is preferable because the $udd$ and $uude$ types correspond to the (SUSY) neutron and hydrogen atom, and so have natural physical interpretations.

To show that the resulting map $\pi \colon M_3 \to Y$ is birational our method is to divide and conquer:  starting with our base set of three GIO types, with image $Y = Y_3$, we add a single GIO type to our base set and consider the projection map from the image 
$Y_4$ to $Y_3 = Y$, and show that this is birational.
Then we iterate the process, adding one GIO type at a time as in Figure~\ref{ProjectionTower}, showing each of these
projections is birational, and consequently so is their composition.
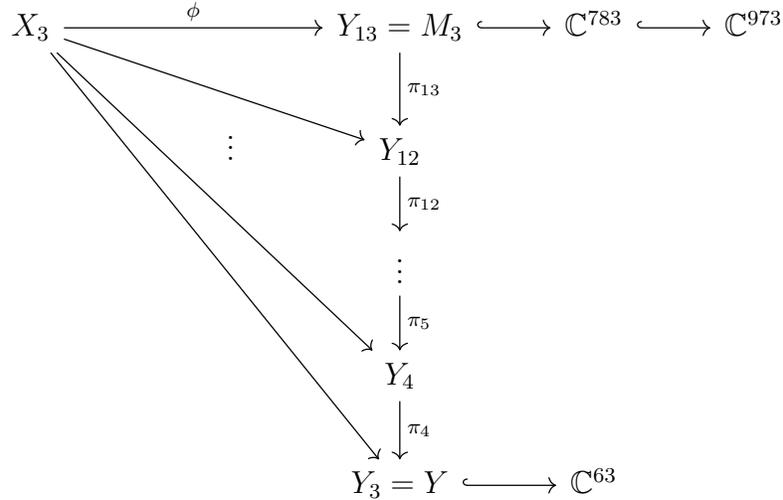
\begin{figure}[ht]
\begin{center}
  \begin{tikzcd}
    X_3 \arrow[rrr, "\phi"] \arrow[drrr]
    \arrow[dddrrr]
    \arrow[ddddrrr]
    & & & Y_{13} = M_3 \arrow[d, "\pi_{13}"] \arrow[r, hook] & \CC^{783} \arrow[r, hook] & \CC^{973}\\
      & & \vdots & Y_{12} \arrow[d, "\pi_{12}"] \\
      & & & \vdots \arrow[d, "\pi_5"] \\
      & & & Y_{4} \arrow[d, "\pi_4"] \\
      & & & Y_{3} = Y \arrow[r, hook] & \CC^{63}
    \end{tikzcd}
\end{center}
\caption{The tower of birational projections from $M_3$ to the base $Y$.}
\label{ProjectionTower}
\end{figure}

Choosing the order of this sequence of projections is a delicate process, because a poor choice at one step makes the determination of birationality of that step computationally exorbitant.
To overcome this, we will use a multigrading on the GIO types to identify the order in which to add the GIO types.

There is a natural $\ZZ^7$ grading on the polynomial ring $R$ with $\{u,d,L,Q,e,H,\barH\}$
corresponding to $\{{\bf e_1}, \ldots, {\bf e_7}\}$.
Each of the $28$ GIO types is homogeneous with respect to this grading. 
However, on $X_3$, the $H_\alpha$ and $\barH_\alpha$ vanish, so we have a $\ZZ^5$ (rather than $\ZZ^7$) grading.
An additional constraint is that the equations defining $X_3$ are not graded, as two of the equations are of the form $\sum dQ + \sum Le$. 
So we further restrict from the $\ZZ^5$-grading to a $\ZZ^4$-grading as in Table~\ref{tbl:three} below.

\begin{table}[h!!!]
\begin{center}{\small
\vspace{-.1in}
\begin{tabular}{|c|c|c|c|c|}\hline
$u$ & $d$ & $L$ & $Q$ & $e$ \\
\hline \hline
1 & 0 & 0 & 0 & 0 \\
0 & 1 & 0 & 0 & 1 \\
0 & 0 & 1 & 0 & -1 \\
0 & 0 & 0 & 1 & 1 \\
\hline
\end{tabular}}
\end{center}\vspace{-.1in}{\caption{\label{tbl:three}
{\textsf{The $\ZZ^4$-grading on component $X_3$.}}
}}
\vspace{-.1in}
\end{table}

By construction, the image $M_3$ of $X_3$ is also $\ZZ^4$-graded.
In fact $I_{X_3}$ and the ideal of every intermediate $Y_i$ in the tower of projections in Figure~\ref{ProjectionTower} are all $\ZZ^4$-graded.

\subsection{Equations for the base space $Y$ using the GIO types $QdL, uude, udd$.}
Recall that $Y = Y_3$ is the image of $X_3$ under the map defined by the set of base GIO types $QdL, udd, uude$ and serves as the foundation for the tower of the projections in Figure~\ref{ProjectionTower}.
We now describe the ideal of $Y = Y_3$, \textit{i.e.}, the polynomial relations on these three types of GIOs, modulo the equations defining $X_3$.

For ease of notation in this section, we will rename the $13$ types of GIOs which are nonvanishing on $X_3$ with variables in Table~\ref{tab:newnames}.
Notice that we avoid certain letters that are used elsewhere in the paper (\textit{e.g.}, $d$, $e$).
Explicitly, we rename the $27$ $QdL$ variables as $\{a_0, \ldots, a_{26}\}$, the nine $udd$ variables as $\{c_0, \ldots, c_{8}\}$, and the $27$ $uude$ variables as $\{h_0, \ldots, h_{26}\}$.
Notice that these $63$ variables are the coordinates on $Y_3 \subset \CC^{63}$.
\vskip -.05in
\begin{table}[h!!!]
\begin{center}{\small
\begin{tabular}{|c||c|c|}\hline
\mbox{Type} & \mbox{Variable} & \mbox{Number} \\
\hline \hline
$QdL$ & $a$  & 27 \\ \hline

$LLe$ & $b$  & 9  \\ \hline
$udd$ & $c$ & 9  \\ \hline
$uuuee$ & $g$ & 6 \\ \hline
$uude$ & $h$  & 27 \\ \hline
$dddLL$ & $i$  & 3 \\ \hline
$QQQL$ & $j$  & 24 \\ \hline
$QuQue$ & $k$  & 108  \\ \hline
$QuLe$ & $m$  & 81 \\ \hline
$QuQd$ & $n$  & 81 \\ \hline

$QQQQu$ & $p$  & 54  \\ \hline

$(QQQ)_4LLLe$ & $q$  & 30 \\ \hline 
$uudQdQd$ & $r$  & 324 \\ \hline %
\end{tabular}}
\end{center}{\caption{\label{tab:newnames}
{\textsf{Renaming the $13$ GIO types for component $X_3$; there are a total of $783$ polynomials in these $13$ types. A table describing all $28$ GIO types is given as an appendix.}}}}
\end{table}
\noindent Now define matrices $M_{a1}, M_{a3}, M_{h1}, M_{h2}, M_{h3}, ,M_C, v_a, v_c, J$ as below:
\[
\begin{array}{{ccl}}
    \vspace{.15in} M_{a1} & := & 
    \left(\!\begin{array}{ccccccccc}
a_{0}&a_{3}&a_{6}&a_{9}&a_{12}&a_{15}&a_{18}&a_{21}&a_{24}\\
a_{1}&a_{4}&a_{7}&a_{10}&a_{13}&a_{16}&a_{19}&a_{22}&a_{25}\\
a_{2}&a_{5}&a_{8}&a_{11}&a_{14}&a_{17}&a_{20}&a_{23}&a_{26}
\end{array}\!\right) \,, \\

 \vspace{.15in} M_{a3} & := & 
    \left(\!\begin{array}{ccccccccc}
a_{0}&a_{1}&a_{2}&a_{9}&a_{10}&a_{11}&a_{18}&a_{19}&a_{20}\\
a_{3}&a_{4}&a_{5}&a_{12}&a_{13}&a_{14}&a_{21}&a_{22}&a_{23}\\
a_{6}&a_{7}&a_{8}&a_{15}&a_{16}&a_{17}&a_{24}&a_{25}&a_{26}
\end{array}\!\right) \,, \\

    \vspace{.15in} M_{h1} & := & 
    \left(\!\begin{array}{ccccccccc}
h_{0}&h_{3}&h_{6}&h_{9}&h_{12}&h_{15}&h_{18}&h_{21}&h_{24}\\
h_{1}&h_{4}&h_{7}&h_{10}&h_{13}&h_{16}&h_{19}&h_{22}&h_{25}\\
h_{2}&h_{5}&h_{8}&h_{11}&h_{14}&h_{17}&h_{20}&h_{23}&h_{26}
\end{array}\!\right) \,, \\

 \vspace{.15in} M_{h2} & := & 
    \left(\!\begin{array}{ccccccccc}
h_{0}&h_{1}&h_{2}&h_{3}&h_{4}&h_{5}&h_{6}&h_{7}&h_{8}\\
h_{9}&h_{10}&h_{11}&h_{12}&h_{13}&h_{14}&h_{15}&h_{16}&h_{17}\\
h_{18}&h_{19}&h_{20}&h_{21}&h_{22}&h_{23}&h_{24}&h_{25}&h_{26}
\end{array}\!\right) \,, \\

 \vspace{.15in} M_{h3} & := & 
    \left(\!\begin{array}{ccccccccc}
h_{0}&h_{1}&h_{2}&h_{9}&h_{10}&h_{11}&h_{18}&h_{19}&h_{20}\\
h_{3}&h_{4}&h_{5}&h_{12}&h_{13}&h_{14}&h_{21}&h_{22}&h_{23}\\
h_{6}&h_{7}&h_{8}&h_{15}&h_{16}&h_{17}&h_{24}&h_{25}&h_{26}
\end{array}\!\right) \,, \\

      \vspace{.15in} M_C & := & 
      \left(\!\begin{array}{cccccccccccccccc}
c_5&0&c_2&0&0&0&c_6&0&c_7&0&0&0&-c_8&0&c_8&0\\
-c_4&0&-c_1&0&0&0&0&c_6&0&0&c_8&0&c_7&c_7&0&0\\
c_3&0&c_0&0&0&0&0&0&0&c_7&0&c_8&-c_6&-c_6&c_6&c_6\\
0&0&0&c_2&c_8&0&-c_3&0&-c_4&0&0&0&-c_5&-c_5&0&c_5\\
0&0&0&-c_1&-c_7&0&0&-c_3&0&0&-c_5&0&c_4&0&-c_4&-c_4\\
0&0&0&c_0&c_6&0&0&0&0&-c_4&0&-c_5&-c_3&0&0&0\\
0&c_5&0&0&0&c_8&c_0&0&c_1&0&0&0&0&c_2&0&0\\
0&-c_4&0&0&0&-c_7&0&c_0&0&0&c_2&0&0&0&c_1&0\\
0&c_3&0&0&0&c_6&0&0&0&c_1&0&c_2&0&0&0&c_0
\end{array}\!\right) \,, \\

\vspace{.15in} v_a & := & \left(\!\begin{array}{ccc}
a_{0}+a_{12}+a_{24}&a_{1}+a_{13}+a_{25}&a_{2}+a_{14}+a_{26}
\end{array}\!\right) \,, \\

\vspace{.15in}v_c & := & \left(\!\begin{array}{ccccccccc}
c_{2}&-c_{1}&c_{0}&c_{5}&-c_{4}&c_{3}&c_{8}&-c_{7}&c_{6}
\end{array}\!\right) \,, \\

J  &:=&
    \left(\!\begin{array}{ccc}
0 & 0&1\\
0&-1&0\\
1&0&0
\end{array}\!\right)
\end{array}
\]

\begin{comp}\label{M3baseIdealgens}
The polynomials in the list below minimally generate the ideal $I_Y$. 
\begin{enumerate}
\item[(1)] The $84$ elements of type $aaa$ are the $3 \times 3$ minors of $M_{a1}$.
\item[(2)] The $108$ elements of type $hh$ are the $2 \times 2$ minors of $M_{h1}$.
\item[(3)] The $9$ elements of type $ah$ are the entries of the product $v_a M_{h1}$.
\item[(4)] The $3$ elements of type $ac$ are the entries of the product $M_{a1} v_c^T$.
\item[(5)] The $48$ elements of type $hc$ are the entries of the product $M_{h1} M_C$.
\item[(6)] The $18$ elements of type $ah^2$: These are quite complex, and we describe them in some detail below:
\end{enumerate}
\begin{itemize}
\item Nine of the $18$ equations of type $(6)$ come from products of the form
 \[
 M_{a1} \cdot \Lambda^2(M_{h2_{ijk}}\cdot J), \mbox{ for }(ijk) \in \{(036),(147), (258)\} \,,
 \]
and we flatten the $3 \times 3$ matrix $\Lambda^2(M_{h2_{ijk}}\cdot J)$ into a $9 \times 1$ column vector.

\item The remaining $9$ equations of type $(6)$ come from products of the form
 \[
 M_{a1} \cdot [\Lambda^2(J \cdot M_{h3_{i,j}}),\Lambda^2(J \cdot M_{h3_{i,k}}),\Lambda^2(J \cdot M_{h3_{k-j+i,k}})], \mbox{ for }(ijk) \in \{(047),(058), (158)\},
 \]
and the matrix of the three concatenated $\Lambda^2$ terms is written as a $9 \times 1$ column vector. 
\end{itemize}
The indexing of the minors of the $\Lambda^2$ matrices is chosen to be consistent with that used by the {\tt MSSM.m2} package.

 \end{comp}

All of these have geometric interpretations.
For example, the $3 \times 3$ entries of $M_a$ define the secant variety to the Segre variety $\Sigma_{2,8}$, while the $48$ entries of type $hc$ can be interpreted as taking Pl\"ucker coordinates on the $18$ maximal minors of the $6 \times 3$ matrix representing $uuuddd$, but excluding the two pure minors $uuu$ and $ddd$. 

\subsection{Birationality of the maps in the tower of GIOs}
We now describe the projections $\pi_i$ and the relations which define the maps.
It is important to note that we are working on $X_3$, and hence can reduce relations modulo the ideal of $X_3$. We begin with a ``proof of concept'' using the most simple relation:
\begin{comp}\label{uuueeinTermsofbase}
Each GIO in the $uuuee$ type may be expressed as a rational function in the $udd$ and $uude$ GIOs.
Recall that we have renamed the six $uuuee$ GIOs as $\{g_0, \ldots, g_5\}$, the $9$ $udd$ GIOs as $\{c_0,\ldots,c_8\}$, and the $27$ $uude$ GIOs as $\{ h_0,\ldots, h_{26}\}$.
In this notation, we obtain that on $X_3$ we have 
\begin{large}
    \[
\begin{array}{ccc}
g_0 & = & \frac{6\,h_{15}h_{21}-6\,h_{12}h_{24}}{c_{8}} \,, \\
g_1 & = & \frac{6\,h_{15}h_{22}-6\,h_{12}h_{25}}{c_{8}} \,, \\
g_2 & = &\frac{6\,h_{16}h_{22}-6\,h_{13}h_{25}}{c_{8}} \,, \\
g_3 & = &\frac{6\,h_{15}h_{23}-6\,h_{12}h_{26}}{c_{8}} \,, \\
g_4 & = &\frac{6\,h_{16}h_{23}-6\,h_{13}h_{26}}{c_{8}} \,, \\
g_5 & = &\frac{6\,h_{17}h_{23}-6\,h_{14}h_{26}}{c_{8}} \,. \\
\end{array}
\]
\end{large}
\end{comp}
The key to finding these relations is to observe that in multidegree $(4,4,-2,4)$ there are monomials of type $g \cdot c$ corresponding to $(uuuee)(udd)$, and of type $h^2$ corresponding to $(udde)(uude)$. 
By exploiting this multigrading, we are able to determine an order for the projections in the tower of Figure~\ref{ProjectionTower} that allows us to compute the relations. The equations displayed above are the simplest example; the expressions as rational functions become increasingly complex for GIO types with more members. 

\begin{comp}\label{GIOtowerOrder}
The main result of this section is the description of the maps appearing in Figure~\ref{ProjectionTower}, which are displayed in Table~\ref{tbl:giofull3}.
The bottommost (hence final) projection in the tower is 
\[
Y_4 \stackrel{\pi_4}{\longrightarrow} Y_3 = Y \,,
\]
and is exactly the map described in Computation~\ref{uuueeinTermsofbase}.
\begin{table}[h!!!]
\begin{center}{\small
\begin{tabular}{|c||c|c|}\hline
\mbox{Map} & \mbox{GIO eliminated} & \mbox{Relations in degrees} \\
\hline \hline
$\pi_{13}$ & $QuQue = k$  & $k=k$, $kj = pm$ \\ \hline
$\pi_{12}$ & $(QQQ)_4LLLe=q$  & $q=ja+jb, qa=jab, qab=jaab$  \\ \hline
$\pi_{11}$ & $QQQQu=p$ &$p=p$, $pi=jac$ \\ \hline
$\pi_{10}$ & $QuLe = m$  & $m = m$, $mcc=cah$ \\ \hline
$\pi_{9}$ & $uudQdQd=r$  & $r=ah$, $rb=aah$ \\ \hline
$\pi_{8}$ & $QuQd=n$  & $n=n$, $ni=aac$  \\ \hline
$\pi_{7}$ & $QQQL = j$  & $ji=aaa$ \\ \hline
$\pi_{6}$ & $dddLL = i$  & $ih=bcc$  \\ \hline
$\pi_{5}$ & $LLE= b$  & $b=a$, $bha=haa$ \\ \hline
$\pi_{4}$ & $uuuee =g$  & $gc=hh$ \\ \hline
\end{tabular}}
\end{center}{\caption{\label{tbl:giofull3}
{\textsf{\!\!\!\!Maps in the tower of Figure~\ref{ProjectionTower}. \!Base GIO types are $QdL=a$, $udd = c$, $uude=h$.}}}}
\label{comp3projections}
\end{table}
\end{comp}

 There are a number of instances where a set of GIOs, when restricted to $X_3$, is not minimal. For example, the type $QuQd$ consists of $81$ GIOs. However, when restricting to $X_3$, nine of these GIOs become linearly dependent. These linear relations are reflected in the table by ``$n=n$'' for $QuQd$; this phenomenon occurs in several other types. Some GIO variables may be expressed linearly in terms of other GIO variables; the ``$b = a$'' for $LLe$ reflects that there are linear dependencies between the $a$ and $b$ variables. Finally, some GIO variables can be expressed as quadrics in other GIO variables, such as ``$q=ja+jb$'' for $(QQQ)_4LLLe$ and ``$r=ah$'' for $uudQdQd$. 
 
\section{Adding neutrinos to the superpotential}\label{sec:six}

Up to this point, we have considered the superpotential without the singlet right handed neutrinos, as in~\eref{eq:renorm}:
\[
\begin{array}{ccc}
W_{\rm minimal}& = & C^0 \sum_{\alpha, \beta} H_\alpha \barH_\beta \eps + \sum_{i,j} C^1_{ij} \sum_{\alpha, \beta, a} Q^i_{a,\alpha} u^j_a H_\beta \eps \nonumber \\
 & &+ \sum_{i,j} C^2_{ij} \sum_{\alpha, \beta, a} Q^i_{a,\alpha} d^j_a \barH_\beta \eps + \sum_{i,j} C^3_{ij} e^i \sum_{\alpha, \beta} L^j_{\alpha} \barH_\beta \eps ~.
 \end{array}
 \]
 We now add in neutrino terms, as in~\eref{eq:neut}
 \[
 \begin{array}{ccc}
 W_{\rm MSSM} & = & W_{\rm minimal} + W_{\rm neutrinos} ~,
 \end{array}
 \]
 where 
 \[
 \begin{array}{ccc}
 W_{\rm neutrinos} & = & \sum_{i,j}C^4_{ij} \nu^i\nu^j+\sum_{i,j} C^5_{ij} \nu^i \sum_{\alpha,\beta} L^j_{\alpha} H_\beta ~. 
\end{array}
\]

\subsection{Relating the superpotentials $W_{\rm minimal}$ and $W_{\rm MSSM}$}
For brevity, let $H\overline{H}, QuH, Qd\overline{H}, L\overline{H}e$ denote the summands in $W_{\rm minimal}$, and $\nu\nu, \nu LH$ the summands in $W_{\rm neutrinos}$. Then the master space equations for $W_{\rm MSSM}$ satisfy
\begin{large}
\begin{equation}\label{MSSMmasterEq}
\begin{array}{ccc}
\frac{\partial W_{\rm MSSM}}{\partial Q}&=&\frac{\partial W_{\rm minimal}}{\partial Q} \,, \\
\frac{\partial W_{\rm MSSM}}{\partial u}&=&\frac{\partial W_{\rm minimal}}{\partial u} \,, \\
\frac{\partial W_{\rm MSSM}}{\partial d}&=&\frac{\partial W_{\rm minimal}}{\partial d} \,, \\
\frac{\partial W_{\rm MSSM}}{\partial e}&=&\frac{\partial W_{\rm minimal}}{\partial e} \,, \\
\frac{\partial W_{\rm MSSM}}{\partial \overline{H}}&=&\frac{\partial W_{\rm minimal}}{\partial \overline{H}} \,, \\
\frac{\partial W_{\rm MSSM}}{\partial H}&=&\frac{\partial W_{\rm minimal}}{\partial H}+\frac{\partial \nu LH}{\partial H} \,, \\
\frac{\partial W_{\rm MSSM}}{\partial L}&=&\frac{\partial W_{\rm minimal}}{\partial L}+\frac{\partial \nu LH}{\partial L} \,, \\
\frac{\partial W_{\rm MSSM}}{\partial \nu}&=&\frac{\partial W_{\rm neutrinos}}{\partial \nu}  \,.
\end{array}
\end{equation}
\end{large}

\subsubsection{Comparison to results on the electroweak sector}
As a quick calibration, we restrict the master space of $W_{\rm MSSM}$ to the electroweak sector, which corresponds to setting (for all indices)
\[
Q = u = d = 0 \mbox{ in }W_{\rm MSSM} ~.
\]
Let $W_{\rm MSSM}^{\rm res}$ denote the specialization above, and compute the Jacobian ideal of $W_{\rm MSSM}^{\rm res}$. 
\begin{comp}\label{sixone}
The restricted master space $W_{\rm MSSM}^{\rm res}$ has two components, defined by the vanishing locus of ideals $P_1$ and $P_2$, where  $($again, all indices present$)$
\[
\begin{array}{cccc}
P_2 &=& \langle Q,u,d, \nu, \overline{H}, H_1+\sum\limits_{i} L^i_1 e^i, H_2+\sum\limits_{i} L^i_2 e^i, I_{2 \times 2}[H|L] \rangle  & = I_2+ \langle Q,u,d, \nu \rangle.\\
P_3 &=& \langle Q,u,d, \nu, \overline{H}, H, \sum\limits_{i} L^i_1 e^i, \sum\limits_{i} L^i_2 e^i \rangle & = I_3+ \langle Q,u,d, \nu \rangle. \\
\end{array}
\]
Here $I_{2 \times 2}[H|L]$ denotes the $2 \times 2$ minors of the $2 \times 4$ matrix whose first column is $[H_1,H_2]^T$ and remaining 3 columns are the $2 \times 3$ matrix of the $L$ variables.
\end{comp}
\noindent On the locus $V(P_2)$, all the GIOs vanish, while on $V(P_3)$, all the GIO types except $LLe$ vanish.
A computation confirms that the image of the restricted master space is the Veronese surface, as shown in \S5 of~\cite{He:2014oha}.

\subsection{The master space and vacuum moduli space for $W_{\rm MSSM}$}
Assuming that the coupling matrices are invertible, note that the last equality in~\eref{MSSMmasterEq},
\[
\frac{\partial W_{\rm MSSM}}{\partial \nu}=\frac{\partial W_{\rm neutrinos}}{\partial \nu} \,, 
\]
shows that each neutrino $\nu^i$ may be expressed as a bihomogeneous quadratic equation in the $L, H$ variables. 

In particular, the $\nu_i$ are expressed parametrically in terms of the $L, H$ variables. There is also a deformation resulting from 
\[
\begin{array}{ccc}
\frac{\partial W_{\rm MSSM}}{\partial H}&=&\frac{\partial W_{\rm minimal}}{\partial H}+\frac{\partial \nu LH}{\partial H} \,, \\
\frac{\partial W_{\rm MSSM}}{\partial L}&=&\frac{\partial W_{\rm minimal}}{\partial L}+\frac{\partial \nu LH}{\partial L} \,. \\
\end{array} 
\]
However, these deformations are fairly simple, because the two partials 
\[
\frac{\partial \nu LH}{\partial H} \mbox{ and } \frac{\partial \nu LH}{\partial L}
\]
yield bihomogeneous polynomials of type $\nu L$ and $\nu H$. 
\begin{comp}
The Jacobian ideal of $W_{\rm MSSM}$ defines the master space of the superpotential with neutrinos.
It has three irreducible components corresponding to prime ideals $J_i$ below; in the notation from \S3, the corresponding ideals are 
\[
\begin{array}{ccc}
I_1 &=& J_1 \,,\\
I_2 & \subseteq & J_2 \,,\\
I_3 & = & J_3 \,.
\end{array}
\]
The codimension of $V(J_2)$ in $V(I_2)$ is two, so $V(J_2)$ is of dimension $25$. Furthermore, the variables $\nu^i$ vanish on each of the components of the master space for $W_{\rm MSSM}$. Hence, the master spaces for $W_{\rm MSSM}$ and $W_{\rm minimal}$ may both be regarded as embedded in the same ambient space. The result above shows that for the corresponding vacuum moduli spaces, the images of two of the three components of the master spaces coincide. 

\noindent The ideal of $J_2$ is minimally generated by $14$ linear forms and $80$ quadrics; using the notation of \S3 we have
\[
J_2 = \langle \overline{H}, u, \nu, I_{2 \times 2}[(M_H)_1|M_Q|M_L], F_1, F_2\rangle \,.
\]
    Note that $J_2$ differs from $I_2$ only in the addition of $\mid M_L$ to (matrix defining) the ideal of $2 \times 2$ minors, and that setting $Q=u=d=0$ in $J_2$ yields the ideal $P_2$.
\end{comp}

\begin{comp}
The image 
\[
V(J_2) \longrightarrow \CC^{976}
\]
under the map determined by the full set of all GIOs is zero. 
\end{comp}
\begin{proof}
Compute the restriction of each GIO type to $V(J_2)$, note they all vanish.
\end{proof}
\noindent 
As $V(J_2) \subset V(I_2)$ we might have naively expected that the image of $V(J_2)$
would reflect the Veronese surface appearing after Computation~\ref{sixone}. But in fact the Veronese surface is the image of the restriction of $V(I_2) \cap V(J_3)$
to the electroweak sector with neutrinos added (\textit{i.e.}, $Q = u = d = \nu = 0$).

In \S\ref{sec:three} we showed that the primary decomposition of the Jacobian ideal of $W_{\rm minimal}$ is the intersection of three \textit{prime} ideals. It turns out that the primary decomposition of the Jacobian ideal of $W_{\rm MSSM}$ is more complicated, however, the radical of this ideal is still an intersection of three prime ideals. It would be interesting to understand the geometric and physical implications of this dichotomy.\vspace{-.1in}
\subsection{Software and future directions}\vspace{-.1in}
Computations were performed using the {\tt Macaulay2} package {\tt MSSM.m2}, which will be part of the next distribution of {\tt Macaulay2}, available at {\tt https://macaulay2.com}. Our work here opens up doors to a number of followup questions on the geometric structure of the vacuum moduli space; we are in the process of extending our computational toolkit to handle other variants of the superpotential.
\vspace{-.1in}
\subsection*{Acknowledgments}
\vspace{-.1in}
Much of this work was done while the last two authors were visitors at the Mathematical Institute at Oxford; we thank the institute for a congenial working environment.
YHH is supported by a Leverhulme Trust Research Project (Grant No.\ RPG-2022-145) and STFC grant ST/J00037X/3.
VJ is supported by the South African Research Chairs Initiative of the Department of Science, Technology, and Innovation and the National Research Foundation, grant 78554.
BDN is supported by NSF PHY-2209903.
HS was supported by NSF DMS-2006410 and a Leverhulme Visiting Professorship at Oxford.
MS was supported by NSF DMS-2001367 and a Simons Fellowship at Oxford.

\appendix
\pagebreak

\section{Appendix: Gauge Invariant Operators in the MSSM}\label{app:GIOs}
\vspace{-.1in}
\begin{table}[ht]
{\small
\begin{tabular}{|c||c|c|c|}\hline
\mbox{Type} & \mbox{Explicit Sum} & \mbox{Index} & \mbox{Number} \\
\hline \hline
$\nu$  & $\nu^i$ & $i=1,2,3$ & 3 \\ \hline
$LH$  & $L^i_\alpha H_\beta \eps$ & $i=1,2,3$ & 3 \\ \hline
$H\barH$ & $H_\alpha \barH_\beta \eps$ & & 1  \\ \hline
$udd$ & $u^i_a d^j_b d^k_c \epsilon^{abc}$ & $i,j=1,2,3$; $k=1,\ldots,j-1$ & 9  \\ \hline
$LLe$ & $L^i_\alpha L^j_\beta e^k \eps$ & $i,k=1,2,3$; $j=1,\ldots,i-1$ & 9  \\ \hline
$QdL$ & $Q^i_{a, \alpha} d^j_a L^k_\beta \eps$ & $i,j,k=1,2,3$ & 27 \\ \hline
$QuH$ & $Q^i_{a, \alpha} u^j_a H_\beta \eps$ & $i,j=1,2,3$ & 9 \\ \hline
$Qd\barH$ & $Q^i_{a, \alpha} d^j_a \barH_\beta \eps$ & $i,j=1,2,3$ & 9 \\ \hline
$L\barH e$ & $L^i_\alpha \barH_\beta \eps e^j$ & $i,j=1,2,3$ & 9 \\ \hline
$QQQL$ & $Q^i_{a, \beta} Q^j_{b, \gamma} Q^k_{c, \alpha} L^l_\delta \epsilon^{abc} \epsilon^{\beta\gamma}\epsilon^{\alpha\delta}$ & $\ba{l} i,j,k,l=1,2,3; i\ne k, j\ne k, \\ j \le i, (i,j,k) \ne (3,2,1) \ea$ & 24 \\ \hline
$QuQd$ & $Q^i_{a, \alpha} u^j_a Q^k_{b, \beta} d^l_b \eps$ & $i,j,k,l=1,2,3$ & 81 \\ \hline
$QuLe$ & $Q^i_{a, \alpha} u^j_a L^k_{\beta} e^l \eps$ & $i,j,k,l=1,2,3$ & 81 \\ \hline
$uude$ & $u^i_a u^j_b d^k_c e^l \epsilon^{abc}$ & $i,j,k,l=1,2,3; j<i$ & 27 \\ \hline
$QQQ\barH$ & $Q^i_{a, \beta} Q^j_{b, \gamma} Q^k_{c, \alpha} \barH_\delta \epsilon^{abc} \epsilon^{\beta\gamma} \epsilon^{\alpha\delta}$ & $\ba{l} i,j,k=1,2,3; i\ne k, j\ne k, \\ j\le i, (i,j,k) \ne (3,2,1) \ea$ & 8 \\ \hline
$Qu\barH e$ & $Q^i_{a, \alpha} u^j_a \barH_\beta e^k \eps$ & $i,j,k =1,2,3$ & 27 \\ \hline
$dddLL$ & $d^i_a d^j_b d^k_c L^m_\alpha L^n_\beta \epsilon^{abc} \epsilon_{ijk} \eps$ & $m,n=1,2,3, n<m$ & 3 \\ \hline
$uuuee$ & $u^i_a u^j_b u^k_c e^m e^n \epsilon^{abc} \epsilon_{ijk}$ & $m,n=1,2,3, n \le m$ & 6 \\ \hline
$QuQue$ & $Q^i_{a, \alpha} u^j_a Q^k_{b, \beta} u^m_b e^n \eps$ & $ \ba{l} i,j,k,m,n=1,2,3; \\ \mbox{antisymmetric}\{ (i,j), (k,m) \}
\ea$ & 108  \\ \hline
$QQQQu$ & $Q^i_{a, \beta} Q^j_{b, \gamma} Q^k_{c, \alpha} Q^m_{f,\delta} u^n_f \epsilon^{abc} \epsilon^{\beta\gamma} \epsilon^{\alpha\delta}$ & $\ba{l} i,j,k,m,n=1,2,3; i\ne k, j\ne k, \\ j\le i, (i,j,k) \ne (3,2,1) \ea$ & 54\\ \hline

%
$dddL\barH$ & $d^i_a d^j_b d^k_c L^m_\alpha \barH_{\beta} \epsilon^{abc}\epsilon_{ijk} \eps$ & $m=1,2,3$ & 3 \\ \hline
$uudQdH$ & $u^i_a u^j_b d^k_c Q^m_{f, \alpha}d^n_f H_\beta \epsilon^{abc} \eps$ & $i,j,k,m,n=1,2,3; j<i$ & 81 \\ \hline
$(QQQ)_4LLH$  & $(QQQ)_4^{\alpha\beta\gamma} L^m_\alpha L^n_\beta H_\gamma$ & $m,n=1,2,3, n \le m$ & 6 \\ \hline
$(QQQ)_4LH\barH$ & $(QQQ)_4^{\alpha\beta\gamma} L^m_\alpha H_\beta \barH_\gamma$ & $m=1,2,3$ & 3 \\ \hline
$(QQQ)_4H\barH\barH$ & $(QQQ)_4^{\alpha\beta\gamma} H_\alpha \barH_\beta \barH_\gamma$ & & 1 \\ \hline
$(QQQ)_4LLLe$ & $(QQQ)_4^{\alpha\beta\gamma} L^m_\alpha L^n_\beta L^p_\gamma e^q$ & $m,n,p,q=1,2,3, n\le m,p \le n$ & 30 \\ \hline %
$uudQdQd$ & $u^i_a u^j_b d^k_c Q^m_{f, \alpha}d^n_f Q^p_{g, \beta} d^q_g \epsilon^{abc} \eps$ & $\ba{l} i,j,k,m,n,p,q=1,2,3; \\ j<i, \mbox{antisymmetric}\{(m,n), (p,q)\} \ea$ & 324 \\ \hline
$(QQQ)_4LL\barH e$ & $(QQQ)_4^{\alpha\beta\gamma} L^m_\alpha L^n_\beta \barH_\gamma e^p$ & $m,n,p = 1,2,3, n \le m$ & 18 \\ \hline %
$(QQQ)_4L\barH\barH e$ & $(QQQ)_4^{\alpha\beta\gamma} L^m_\alpha \barH_\beta \barH_\gamma e^n$ & $m,n=1,2,3$ &  9 \\ \hline
$(QQQ)_4\barH\barH\barH e$ & $(QQQ)_4^{\alpha\beta\gamma} \barH_\alpha \barH_\beta \barH_\gamma e^m$ & $m=1,2,3$ &  3 \\ \hline
\end{tabular}}
\caption{
\textsf{The 
$976$ generators of gauge invariant operators for the MSSM. For $QQQQu$, a na\"{\i}ve count gives $72$, but there are $18$ linear relations, reducing the number of independent operators.}
}
\label{tab:giofull}
\vspace{-1in}
\end{table}
\pagebreak
{\small
\bibliographystyle{JHEP}
\bibliography{refs}
}

\end{document}